\documentclass[journal]{IEEEtran}
\usepackage[utf8]{inputenc} 
\usepackage[T1]{fontenc}
\usepackage{url}
\usepackage{ifthen}
\usepackage{cite}
\usepackage{hyperref}
\usepackage{overpic}
\usepackage[cmex10]{amsmath} 

\usepackage{amsthm}
\newtheorem{example}{Example}

\DeclareSymbolFont{extraup}{U}{zavm}{m}{n}
\DeclareMathSymbol{\varheart}{\mathalpha}{extraup}{86}
\DeclareMathSymbol{\vardiamond}{\mathalpha}{extraup}{87}

\usepackage{mathtools}
\mathtoolsset{showonlyrefs}

\newcommand{\comp}{\mathsf{c}}

\newcommand{\red}[1]{\textcolor{black}{#1}}

\usepackage{preamble_IEEE}

\begin{document}


\title{Ensemble-Tight Second-Order Asymptotics and  Exponents  for Guessing-Based Decoding with Abandonment} 


\author{Vincent Y. F. Tan, \emph{Senior Member,~IEEE}  $\qquad$ Hamdi Joudeh,~\emph{Member,~IEEE}
\thanks{V. Y. F. Tan is with the Department of Mathematics and Department of Electrical and Computer Engineering,
National University of Singapore, 119077, Singapore (e-mail: vtan@nus.edu.sg). H.~Joudeh is with the Department of Electrical Engineering,  Eindhoven University of Technology (e-mail: h.joudeh@tue.nl). 

V.~Y.~F.~Tan is supported by a Singapore Ministry of Education (MOE) AcRF Tier~2 grant (A-8000423-00-00) and two AcRF Tier~1 grants (A-8000980-00-00 and A-8002934-00-00). H.~Joudeh is supported by a European Research Council (ERC) Starting Grant (No. 101116550).

This paper was presented in part at ISIT  2025.
}}

\maketitle

\begin{abstract}
This paper considers guessing-based decoders with abandonment for discrete memoryless channels in which all codewords have the same composition. This class of decoders rank-orders all input sequences in the codebook's composition class  from ``closest'' to ``farthest'' from the channel output and then queries them sequentially in that order for codebook membership. Decoding terminates when a codeword is encountered or when a predetermined number of guesses is reached, and decoding is abandoned. We  derive ensemble-tight first-order asymptotics for the code rate and abandonment rate, which shows that guessing-based decoding is more efficient than conventional testing-based decoding whenever the capacity of the channel  exceeds  half the entropy of the capacity-achieving input distribution.  The main focus of this paper is on refined asymptotics, specifically, second-order asymptotics, error exponents, and strong converse exponents. The optimal second-order region is characterized in terms of the minimum of the second-order code and abandonment rates. The error (resp.\ strong converse) exponent  is characterized in terms of the minimum (resp.\ maximum) of the usual channel coding exponent and an abandonment exponent, which turns out to be a special case of the exponent of conditional almost-lossless source coding.
\end{abstract}
  
\section{Introduction}
\label{sec:introduction}
Consider a discrete memoryless channel (DMC) with finite input and output alphabets $\mathcal{X}$ and $\mathcal{Y}$, respectively, and with  transition probability law $W: \mathcal{X} \to \mathcal{Y}$.
A transmitter wishes to communicate a message from the set $[M_n] := \{1,2,\ldots,M_n\}$ to a receiver 
over $n$ uses of the DMC.
To communicate $w \in [M_n]$, the transmitter sends a corresponding length-$n$ codeword $\bm{x}(w) \in \mathcal{X}^n$. 
We assume that the employed codebook $\mathcal{C}_n :=  \{ \bm{x}(1), \bm{x}(2), \ldots, \bm{x}(M_n) \}$ has fixed composition, i.e., all codewords come from the same type class.

In proving coding theorems for the DMC, a standard approach is to analyze a randomly generated codebook $\mathcal{C}_n$ in tandem with a decoding procedure that, given a channel output $\bm{y} \in \mathcal{Y}^n$, tests all codewords in $\mathcal{C}_n$ and selects one which maximizes some decoding metric. 
Notable special cases include standard maximum likelihood (ML) decoding \cite{Gallager1968} and maximum mutual information (MMI) decoding \cite{Goppa}. 
We shall refer to this conventional approach as \emph{testing-based decoding}, due to its roots in $M_n$-ary hypothesis testing.

In this paper, we follow a different approach to channel decoding based on guessing (or  guesswork) \cite{Massey1994,Arikan1996}, which works as follows.
Given a channel output $\bm{y}$, the decoder rank-orders all sequences $\bm{x}$, codewords or otherwise, in the codebook's type class from ``closest'' to ``farthest'' from $\bm{y}$ with respect to some ``distance'' measure, and then queries them sequentially in that order for codebook membership. 
Decoding stops at the first instance a codeword is encountered, or if a predetermined maximum number of guesses $m_n$ is reached. In the latter case, search is abandoned and an error is declared.
We shall refer to this approach as \emph{guessing-based decoding with abandonment}, or \emph{guessing-based decoding} for short. 

In a somewhat loose sense, the above procedure searches for the closest codeword to the output $\bm{y}$, contained in an input-space preimage of some neighborhood of $\bm{y}$.
Intuitively, if the DMC is not too noisy, then it suffices to choose a maximum number of guesses $m_n$ that is much smaller than the number of codewords $M_n$; rendering guessing-based decoding more efficient than testing-based decoding in terms of search complexity (i.e., the number of queried sequences).\footnote{An implicit underlying assumption here is that there are efficient routines for rank-ordering and codebook membership checking.}
This intuition is made concrete later on in this paper, and is done so via three asymptotic regimes—second-order asymptotics~\cite{Tan2014, Hayashi09}, error exponents (or reliability function) \cite{Gallager1968,Csi97}, and strong converse exponents~\cite{arimoto73,DK79}.

The guessing-based decoding approach we describe above is directly inspired by guessing random additive noise decoding (GRAND), proposed by Duffy \emph{et al.} \cite{Duffy2019} for discrete modulo-sum channels, where the output is a modulo-sum of the input and independent noise (all drawn from the same alphabet).\footnote{Note that the idea of decoding by guessing the noise pattern has a long history predating GRAND, e.g., syndrome decoding of linear codes.}
GRAND rank-orders noise sequences from most to least likely and subtracts them sequentially from the channel output, at each stage checking if the result is a codeword, and stopping at the first encountered codeword or if a maximum number of guesses is reached.
Various extensions of GRAND have appeared since then, e.g., \cite{Solomon2020,Duffy2021,Duffy2022,An2022,An2023,Joudeh2024}, mostly focusing on Gaussian channels and practical implementations. The original discrete setting in \cite{Duffy2019} is most relevant to our current work.

Viewed in light of the guessing-based decoding procedure described earlier, and under memoryless noise, GRAND can be seen as a special case that exploits the structure of modulo-sum channels, i.e., an ordered list of noise sequences induces an ordered list of input sequences by subtraction from the channel output sequence.
A main motivation of our current paper is to extend the GRAND paradigm to general 
DMCs without the additive structure. 
\subsection{Contributions} 
For a channel code that employs guessing-based decoding with abandonment, define the \emph{code rate} and the  \emph{abandonment rate} respectively by 
\begin{equation*}
    R_n := \frac{1}{n} \log M_n \ \ \text{and} \ \ r_n := \frac{1}{n} \log m_n.
\end{equation*}
For any given decoding error probability, it is clearly desirable to make $R_n$ as large as possible (efficiency) while keeping $r_n$ as small as possible (complexity).
Our three main contributions in this paper are as follows. 
\begin{enumerate} 
    \item We characterize the asymptotic first- and second-order fundamental limits for $R_n$ and $r_n$ for non-vanishing error probabilities. These characterizations, as well as those involving exponents below, involve  direct (achievability) and ensemble converse parts. The latter is deemed appropriate in this setting since we consider a \emph{pre-specified} and \emph{fixed} class of coding schemes with random constant composition codebooks and guessing-based decoders. The first-order result suggests that guessing-based decoding is more efficient than testing-based decoding whenever the capacity $C(W)$ of the DMC $W$ exceeds half the entropy of the capacity-achieving input distribution (CAID) $H(P_X)$. We show that the optimal second-order region is characterized in terms of the {\em minimum} of the second-order code and abandonment rates.
    \item We also characterize ensemble-tight error exponents for guessing-based decoding with abandonment. This regime is of interest for code rates $R=R_n$ {\em below} the channel capacity $C(W)$ and abandonment rates $r=r_n$ above the conditional entropy $H(X|Y)$ (evaluated at the CAID). Denoting the usual random coding exponent for constant composition codes~\cite{Csi97} at rate $R$ as $E_{\mathrm{r}}(R,P_X)$, where $P_X$ here is the composition limiting distribution, our result shows that the error exponent for our setting when $R$ is larger than the critical rate $R_{\mathrm{cr}}$ and $r$ is smaller than $H(P_X)- R_{\mathrm{cr}}$ is   $E_{\mathrm{r}}(\max\{R , H(P_X)-r\},P_X)$, underscoring again that the performance of the system is dominated by either the code rate $R$ being large or the abandonment rate~$r $ being small. 
    \item Finally, considering code rates {\em above} $C(W)$ or abandonment rates {\em below} $H(X|Y)$, we derive ensemble-tight strong converse exponents for the decoder under consideration. Denoting a sphere-packing variant of the strong converse exponent~\cite{DK79, arimoto73, Omura1975} at rate $R$ and a composition limiting distribution $P_X$ as $K_{\mathrm{sp}}(R,P_X)$, analogous to the error exponent setting, our result shows that the  ensemble strong converse exponent is $K_{\mathrm{sp}}( \max\{R, H(P_X)-r\},P_X)$. 
\end{enumerate}
Taken together, the above results suggest that either the code rate or the abandonment rate dominates the behavior of the ensemble error probability in the three asymptotic regimes considered. 
To the best of the authors' knowledge, this is the first work that conducts second-order and comprehensive exponent  analyses, with ensemble converses, for guessing-based decoding with abandonment~\cite{Duffy2019}.

\red{Before we proceed, it is worth noting that while the original work by Duffy \emph{et al.}~\cite{Duffy2019} also treats error and success probability exponents for guessing-based decoding, the analysis there is only applicable to the class of {\em symmetric  additive} channels, and is limited to achievability results with no converses (see also~\cite{Joudeh2024}).
Moreover, the analysis of abandonment in \cite{Duffy2019} is focused on the case $r = H(X|Y)$, and does not consider the whole range of abandonment rates. On the flip side, the tools we use in the current paper are applicable mainly to memoryless channels, while those used in \cite{Duffy2019, Joudeh2024} apply more broadly to channels with memory (albeit in the restricted class of modulo-sum channels).}
\subsection{Notation}
 Random variables (e.g., $X$) and their realizations (e.g.,~$x$) are in upper and lower case, respectively. Their supports are denoted by calligraphic font (e.g.,~$\mathcal{X}$) and the  cardinality of the finite set $\mathcal{X}$ is denoted as $|\mathcal{X} |$. 
Let  $\bm{X}$, whose length should be clear from the context, be the vector of random variables $(X_1, X_2, \ldots, X_n)$ and $\bm{x} = (x_1, x_2, \ldots, x_n)$ denote a realization of $\bm{X}$. The set of  probability mass functions (PMFs) supported on a finite alphabet $\mathcal{X}$ is denoted as $\mathcal{P}(\mathcal{X})$. The set of all conditional distributions with input and output alphabets $\mathcal{X}$ and  $\mathcal{Y}$ respectively is denoted as $\mathcal{P}(\mathcal{Y}|\mathcal{X})$. The joint and output distributions induced by a PMF  $P \in \mathcal{P}(\mathcal{X})$ and a channel (DMC) $W \in \mathcal{P}(\mathcal{Y}|\mathcal{X})$ is $P\times W  \in \mathcal{P}(\mathcal{X}\times \mathcal{Y})$ and $PW\in\mathcal{P}(\mathcal{Y})$ respectively. 

The {\em type}, {\em composition}, or {\em empirical distribution} of a length-$n$ sequence $\bm{x}$ is the PMF
$\{ P_{\bm{x}}(a)=\frac{1}{n}\sum_{i=1}^n\mathbbm{1} [x_i=a]:a\in\mathcal{X}\}$.
The set of all types formed from length-$n$ sequences,  called {\em $n$-types}, is denoted as $\mathcal{P}_n(\mathcal{X})$. The set of all sequences with type $P \in \mathcal{P}_n(\mathcal{X})$ is the {\em type class} $\mathcal{T}_P$.  Given a sequence $\bm{x}\in \mathcal{T}_P$, the set of all sequences $\bm{y} \in \mathcal{Y}^n$ such that $( \bm{x}, \bm{y} )$ has joint type $P\times V$ is the {\em $V$-shell}, denoted as $\mathcal{T}_V(\bm{x})$. The conditional distribution $V$ is also known as the {\em conditional type} of $\bm{y}$ given $\bm{x}$. Conditional types of sequences $\bm{x}$ given a sequence $\bm{y}$, which arise when analyzing reverse channels, are denoted by $V_{X|Y}$ for clarity. We let $\mathcal{V}_n(\mathcal{Y};P)$  be the set of $V \in \mathcal{P}(\mathcal{Y}|\mathcal{X})$ for which the $V$-shell of a sequence of type $P \in \mathcal{P}_n(\mathcal{X})$ is non-empty. 

We use standard notation for information-theoretic quantities such as conditional entropy $H(Y|X)$ and mutual information $I(X;Y)$, which we also write in terms of the underlying distributions as $H(W|P)$ and $I(P,W)$ if $(X,Y)\sim P\times W$. To emphasize the distributions of the random variables, we sometimes also use a subscript on the information-theoretic quantity, e.g., $H_{P_X\times W}(X|Y)$. Given a DMC $W \in \mathcal{P} ( \mathcal{Y}| \mathcal{X})$, the {\em information capacity} is 
$C(W) := \max_{ P \in \mathcal{P}(\mathcal{X})}I(P,W)$.
The set of CAIDs is
$\Pi(W) := \left\{ P\in \mathcal{P}(\mathcal{X}) : I(P,W)=C(W) \right\}$. 
The {\em conditional information variance} is 
\begin{align}
    V(P,W) := \E \left[ \mathrm{Var} \Big(\log\frac{W(\cdot |X)}{PW(\cdot )}\,\Big|\, X \Big)  \right] .
\end{align}
The outer expectation is taken with respect to (w.r.t.) $X\sim P$, while the inner variance is taken w.r.t.\ $Y|X\sim W(\cdot|X)$.
Following the convention in~\cite{TomTan13}, we define the {\em $\varepsilon$-dispersion} of $W$ as 
\begin{align}
    V_\varepsilon(W) := \left\{ \begin{array}{cc}
         V_{\min}(W)& \mbox{if }\varepsilon<1/2 \\
         V_{\max}(W)& \mbox{if }\varepsilon\ge 1/2
    \end{array} \right.  ,
\end{align}
where $V_{\min}(W):= \min_{P\in\Pi(W)}V(P,W)$ and $V_{\max}(W):= \max_{P\in\Pi(W)}V(P,W)$.


The {\em empirical conditional entropy} of $\bm{y} \in \mathcal{Y}^n$ given $\bm{x} \in \mathcal{X}^n$ is written as 
$\hat{H}(\bm{y}|\bm{x}) := H(V | P)$ where $(\bm{x},\bm{y})\in \mathcal{T}_{P\times V}$. 
The {\em empirical mutual information} between $\bm{x}$ and $\bm{y}$ is 
    $\hat{I}( \bm{x}\wedge  \bm{y}) := I(P,W)$ where $P_{ \bm{x},\bm{y} }=P\times W$.
Given $s,t\in\mathbb{R}$, we write $s\wedge t$ and  $\min\{s,t\}$ interchangeably. The uniform distribution over a set $\mathcal{A}$ is $\mathrm{Unif}(\mathcal{A})$. Finally, the   complementary cumulative distribution function of a standard Gaussian is denoted as 
\begin{align}
    \mathrm{Q}(z) = \int_{z}^\infty\frac{1}{\sqrt{2\pi}}\e^{-u^2/2}\, \mathrm{d}u. 
\end{align}
\section{Setting and Preliminaries}
We consider a DMC as described in Section \ref{sec:introduction}, and a coding scheme with a constant composition codebook $\mathcal{C}_n$ of $M_n$ length-$n$ codewords and a guessing-based decoder with at most $m_n$ guesses.
We denote the composition (i.e., type) of $\mathcal{C}_n$ by $P_X^{(n)} \in \mathcal{P}_n(\mathcal{X})$.
Given $\bm{y}$, the decoder rank-orders all $\bm{x} \in \mathcal{T}_{P_X^{(n)}}$ and queries them one by one for membership in $\mathcal{C}_n$, until a codeword if found or guessing is abandoned.

We are interested in universal decoding rules that do not explicitly depend on the channel transition law, and hence we propose to rank input sequences in an increasing order of their empirical conditional entropy $\hat{H}(\bm{x} | \bm{y})$, breaking ties arbitrarily.
Let $\hat{G}:\mathcal{T}_{P_X^{(n)}} \times \mathcal{Y}^n \to [1:  | \mathcal{T}_{P_X^{(n)}} |]$ be a rank function such that $\hat{G}(\bm{x}|\bm{y})$ is the order in which $\bm{x}$ is guessed given that $\bm{y}$ is received. 
The rank function of choice satisfies
\begin{align}
\label{eq:rank_function_condition_1}
    \hat{G}(\bm{x}|\bm{y}) < \hat{G}(\bm{x}'|\bm{y}) & \implies \hat{H}(\bm{x}|\bm{y}) \leq \hat{H}(\bm{x}'|\bm{y}) \\
\label{eq:rank_function_condition_2}    
     \hat{H}(\bm{x}|\bm{y}) < \hat{H}(\bm{x}'|\bm{y})  & \implies \hat{G}(\bm{x}|\bm{y}) < \hat{G}(\bm{x}'|\bm{y})
\end{align}
for any pair of input sequences $\bm{x}$ and $\bm{x}'$ from  $\mathcal{T}_{P_X^{(n)}}$.

With the above rank function, if guessing is allowed to proceed without abandonment, by choosing $m_n = |\mathcal{T}_{P_X^{(n)}}|$, then it is not difficult to verify that the decoded codeword minimizes the empirical conditional entropy, which is known as minimum entropy (ME) decoding. Since we are using constant composition codes, this is equivalent to maximizing the empirical mutual information $\hat{I}(\bm{x}\wedge\bm{y})$, known as maximum-mutual information (MMI) decoding \cite{Goppa}.
MMI decoding \cite{Goppa} is known to achieve capacity, the reliability function above the critical rate \cite{Csi97}, and optimal second-order asymptotics when used in combination with an appropriately chosen sequence of constant composition codes \cite{wang2011}.
The main question we wish to answer is that of how small can we make the maximum number of guesses $m_n$ (or the abandonment rate $r_n$), while still achieving near-MMI performances? 
To this end, we seek to characterize the set of achievable code-abandonment rate pairs $(R_n,r_n)$ in the asymptotic regime of large $n$. What we mean by this will be made precise in what follows.
\subsection{Decoding Rule and Ensemble Error Probability}
\label{sec:error_prb_rates}
To simplify the error analysis, we shall assume that ties in the empirical conditional entropy are also declared as errors. Such events can be detected by slightly modifying the guessing procedure as follows. If a codeword $\bm{x}$ is encountered, then guessing continues until all \emph{equivalent} conditional type classes are exhausted, and the decoder declares an error if at least one more codeword is found.
By ``equivalent'' conditional type classes, we mean all $\mathcal{T}_{V_{X|Y}}(\bm{y})$ such that the conditional entropies $H(V_{X|Y}|\hat{P}_{\bm{y}})$ are equal.
Moreover, abandonment is not allowed to occur before equivalent conditional type classes are exhausted. 
With this in mind, it helps to define the slightly modified rank function
\begin{equation}
\label{eq:modified_rank}
    G(\bm{x} | \bm{y}) = \sum_{V_{X|Y}:H(V_{X|Y}|\hat{P}_{\bm{y}})\leq \hat{H}(\bm{x}|\bm{y})} \big| \mathcal{T}_{V_{X|Y}}(\bm{y}) \big|.
\end{equation}
Note that $G(\cdot|\bm{y})$ is not one-to-one; this   is in contrast to $\hat{G}(\cdot|\bm{y})$.

Given that $\bm{x}(w) \in \mathcal{C}_n$ is sent, and with the modified rank function in \eqref{eq:modified_rank}, an error is made if the channel produces an output $\bm{y}$ such that $G(\bm{x}(\bar{w}) | \bm{y}) \leq G(\bm{x}(w) | \bm{y})$ for some $\bar{w} \in [M_n]$ other than $w$, or $G(\bm{x}(w)|\bm{y}) > m_n $.
Given the codebook $\mathcal{C}_n$, let $\varepsilon_n(\mathcal{C}_n)$ be the error probability of the guessing-based decoder averaged over messages.
An {\em $(n,R_n,r_n)$-code} is specified by a codebook $\mathcal{C}_n$ of length-$n$ codewords and rate $R_n$, and a guessing-based decoder with abandonment rate $r_n$. This code has corresponding error probability $\varepsilon_n(\mathcal{C}_n)$.

We consider random coding ensembles in which each codeword $\bm{X}(w)$, $w\in [M_n]$, is drawn uniformly at random from a type class $\mathcal{T}_{P_X^{(n)}}$ independent of other codewords. Furthermore, the sequence of types $\{P_X^{(n)}\}_{n\in\mathbb{N}}$ is assumed to converge to a particular distribution  $P_X\in\mathcal{P}(\mathcal{X})$ (oftentimes a CAID). In this case, the codebook $\mathsf{C}_n$ is random and the  {\em ensemble average error probability} 
\begin{align}
\varepsilon_n := \E[ \varepsilon_n(\mathsf{C}_n)] = \sum_{\mathcal{C}_n}\Prb[ \mathsf{C}_n=\mathcal{C}_n ]\varepsilon_n(\mathcal{C}_n)    
\end{align}
 is computed by taking a further average over $\mathcal{C}_n$. We call such a random code ensemble  an {\em $(n,R_n,r_n, P_X^{(n)})$-code}. 


\subsection{Alternative Ensembles and Decoding Rules}\label{sec:alt_ensembles}
 For the remainder of this paper, we focus exclusively on:
\begin{enumerate}
    \item[(i)] Ensembles of codes with fixed composition $P_X^{(n)}$;
    \item[(ii)] The use of a guessing-based decoder with abandonment, employing the modified rank function~\eqref{eq:modified_rank}.
\end{enumerate}
 Before proceeding, one might wonder how the results would be impacted by adopting variations of our setting such as the following:
\begin{enumerate}
    \item[(I)] Ensembles of i.i.d.\ random codes with distribution $P_X$;
    \item[(II)] The use of a guessing-based decoder with abandonment based on a non-universal criterion, e.g., maximum-likelihood or maximum a posteriori decoding.
\end{enumerate}
Given the well-known similarities between (i) and (I) across various asymptotic regimes (such as the second-order regime), as well as between (ii) and (II), we expect that the theorems to follow will remain valid if we replace (i) with (I), (ii) with (II), or both.
\section{Ensemble-Tight First-Order Rates}
Here we present a first-order rate result for guessing-based decoding as described in previous sections. While this may be of interest in its own right, our main goal is to derive a first-order \emph{capacity} pair, around which we ``center'' the second-order rate analysis in the next section. In this and next section, $P_X$ always denotes a CAID.  To commence, we formally define the notion of first-order rates. 
\begin{definition}
\red{The pair of first-order rates $(R,r) \in \mathbb{R}_+^2$ is  {\em $(\varepsilon,P_X)$-achievable} if  the sequence of $(n,R_n,r_n,  P_X^{(n)})$-codes  satisfies (i) $P_X^{(n)}$ converges to the CAID $P_X$ and (ii) the code and abandonment rates and random-coding error probabilities $\{\varepsilon_n\}_{n\in\mathbb{N}}$   satisfy
\begin{equation*}
    \liminf_{n \to \infty} R_n  \geq R, \ \  \limsup_{n \to \infty} r_n  \leq r, \  \text{and} \  \limsup_{n \to \infty} \varepsilon_{n}  \le \varepsilon.
\end{equation*}
The closure of the set of all $(\varepsilon, P_X)$-achievable first-order rates for the DMC $W$ is denoted as $\mathcal{R}_{\varepsilon}^*(P_X,W)$. }
\end{definition}

\begin{theorem} \label{thm:first-order} 
Let $P_X\in \Pi(W)$. Then, for all $\varepsilon \in [0,1)$,
    \begin{align}
     \mathcal{R}_\varepsilon^*(P_X,W)\! =\!\bigg\{ (R , r)  :\parbox[c]{1.78in}{$\qquad\qquad\quad\;\;\, 0\le R\le C(W)$, \vspace{0.04 in}\\  $H_{P_X\times W}(X|Y)\le r\le H(P_X)$} \!\!  \bigg\}.
     \end{align}
\end{theorem}
Theorem \ref{thm:first-order} implies that the optimal rate pair is  $(R,r) =( I(P_X,W), H(P_X)-I(P_X,W))$, where $P_X \in \Pi(W)$.
Thus, when operating at these optimal rates, guessing-based decoding is more efficient than testing-based decoding whenever $I(P_X,W) > \frac{1}{2}H(P_X)$; this has also been observed in \cite{Joudeh2024}, albeit in the restricted context of   modulo-sum channels.

\subsection{Proof of Theorem \ref{thm:first-order} }
Due to symmetry in the random codebook ensemble, $\varepsilon_n$ remains unchanged if we condition on the event that message~$1$ is sent.  Therefore, we can write $\varepsilon_n = \Prb \left[ \mathcal{E}_1 \cup \mathcal{A}_1\right]$, where the incorrect decoding event $\mathcal{E}_1$ and the abandonment event $\mathcal{A}_1$ for message  $1$ are respectively defined as 
\begin{align}
\mathcal{E}_1 & := \bigcup_{w \neq 1}\left\{ G(\bm{X}(w) | \bm{Y}) \leq G(\bm{X}(1)|\bm{Y}) \right\}\quad\mbox{and} \label{eqn:def_E1} \\*
\mathcal{A}_1  & := \left\{ G(\bm{X}(1)|\bm{Y}) > \e^{nr_n} \right\}. \label{eqn:defA1}
\end{align}
Therefore, $\varepsilon_n$ is bounded from above and below as follows
\begin{equation}
\label{eq:error_union_max_bounds}
\max \left\{\Prb \left[ \mathcal{E}_1 \right] , \Prb \left[ \mathcal{A}_1 \right] \right\}   \leq \varepsilon_n \leq \Prb \left[ \mathcal{E}_1 \right] +    \Prb \left[ \mathcal{A}_1 \right].
\end{equation}
Before we proceed, it is useful to note from~\eqref{eq:modified_rank} that the modified rank function $G(\bm{x}|\bm{y})$ satisfies
\begin{equation}
\label{eq:G_bounds}
 \e^{n \hat{H}(\bm{x}|\bm{y})} \leq   G(\bm{x}|\bm{y}) \leq (n+1)^{|\mathcal{X}| |\mathcal{Y}|} \e^{n \hat{H}(\bm{x}|\bm{y})}.
\end{equation}

Now for the achievability part, consider a sequence of $(n,R,r,P_X^{(n)})$-codes for which $R = I(P_X,W) - \delta$ and $r = H_{P_X \times W}(X|Y) + \delta$ for some $\delta > 0$.
Using \eqref{eq:G_bounds}, the probability of incorrect decoding is bounded above as 
\begin{equation}
  \Prb \left[ \mathcal{E}_1 \right] \leq \Prb \Bigg[ \bigcup_{w \ne 1}\big\{ \hat{H}(\bm{X}(w) | \bm{Y}) \leq \hat{H}(\bm{X}(1)|\bm{Y}) + \delta_n \big\}\Bigg] 
\end{equation}
where $\delta_n = \frac{|\mathcal{X}||\mathcal{Y}|}{n} \log (n+1)$.
\red{Hence $\Prb \left[ \mathcal{E}_1 \right]$ is 
bounded (up to the presence of the remainder term $\delta_n$) by the ensemble average error probability} achieved by the  ME (therefore the MMI) decoder. From standard analysis and results, we know that $\Prb \left[ \mathcal{E}_1 \right]$ tends to zero as $n$ grows large. 
On the other hand, using the upper bound in~\eqref{eq:G_bounds}, the probability of abandonment is bounded as
\begin{align}
\Prb \left[ \mathcal{A}_1 \right] 
& \leq \Prb \big[ \hat{H}(\bm{X}|\bm{Y}) \geq r  - \delta_n \big]  \\
& = \Prb \big[ \hat{I}(\bm{X} \wedge \bm{Y}) \leq I(P_X,W) - \delta + \delta_n \big] \label{eqn:PA1_UB}
\end{align}
Using properties of conditional types \cite{Csi97}, we can show that the probability in \eqref{eqn:PA1_UB} tends to $0$ as $n$ grows large (see Section~\ref{subsec:achievability_ee}).
As a result, $\varepsilon_n \leq \Prb \left[ \mathcal{E}_1 \right] +    \Prb \left[ \mathcal{A}_1 \right] \to 0$ as $n \to \infty$.

Next, we show an ensemble strong converse. To this end, consider a sequence of  $(n,R_n,r_n,P_X^{(n)})$-codes with rates that satisfy at least one of the following two conditions:
\begin{align}
    \liminf_{n \to \infty} R_n  &\geq R  = I(P_X,W) + \delta \ \ \text{or}  \label{eqn:R_converse} \\
    \limsup_{n \to \infty} r_n  &\leq r  = H_{P_X \times W} (X|Y) - \delta \label{eqn:r_converse}
\end{align}
for some $\delta > 0$.
If \eqref{eqn:R_converse} holds, then by the strong converse to the channel coding theorem, it holds that $\Prb \left[ \mathcal{E}_1 \right] \to 1$ as $n \to \infty$.
Otherwise, if \eqref{eqn:r_converse} holds, for all $n$ sufficiently large, 
\begin{equation}
r_n \le H_{P_X\times W}(X|Y)-\delta/2  \label{eqn:r_converse2}    
\end{equation}
and we consider the following lower bound on the probability of abandonment
\begin{align}
\Prb \left[ \mathcal{A}_1 \right] 
& \geq \Prb \big[ \hat{H}(\bm{X}|\bm{Y}) \geq r_n \big] \\
& \geq \Prb\big[ \hat{I}(\bm{X} \wedge \bm{Y}) \leq I(P_X,W) + \delta/2 \big] \label{eqn:PA1_LB}
\end{align}
where \eqref{eqn:PA1_LB} is due to \eqref{eqn:r_converse2}. 
The lower bound in \eqref{eqn:PA1_LB} converges to one (see Section \ref{subsec:converse_sce}), and hence 
$\Prb \left[ \mathcal{A}_1 \right] \to 1 $  as $n \to \infty$. This concludes the proof.
\section{Ensemble-Tight Second-Order Rates} \label{sec:second_order}
In this section, we  derive ensemble-tight second-order asymptotics for guessing-based decoding with abandonment. 
Since the error probability is a function of both the code  and abandonment rates, we consider the scenario in which the two rates involved deviate from their first-order fundamental limits derived in Theorem \ref{thm:first-order} by  terms that vanish as $\Theta\big( \frac{1}{\sqrt{n}} \big)$. This motivates the following definition.

\begin{definition} \label{def:2nd_order}
\red{The pair of second-order rates $(s,t) \in\mathbb{R}^2$ is  {\em $(\varepsilon, P_X)$-achievable} if the sequence of $(n,R_n,r_n, P_{X}^{(n)})$-codes satisfies  (i) $P_X^{(n)}$ converges to 
\begin{align}
    P_X  \in \left\{ \begin{array}{cc}
         \argmin_{P\in \Pi(W)}V(P,W)  & \mbox{if } s\wedge t\ge 0  \\
         \argmax_{P\in \Pi(W)}V(P,W) & \mbox{if } s\wedge t < 0 
    \end{array}   \right. \label{eqn:caid}
\end{align}
and (ii) the code and abandonment rates and random-coding error probabilities  $\{\varepsilon_n\}_{n\in\mathbb{N}}$    satisfy
\begin{align}
    \liminf_{n\to\infty}\sqrt{n} \big(R_n-C( W)\big) &\ge -s ,\label{eqn:deviate_rate} \\
    \limsup_{n\to\infty}\sqrt{n} \big(r_n-H_{ P_X \times W} (X|Y) \big) &\le t,\quad\mbox{and}\label{eqn:deviate_abn}\\*
    \limsup_{n\to\infty}\varepsilon_n&\le\varepsilon. \label{eqn:asymp_er}
\end{align}
The closure of the set of all $(\varepsilon,P_X)$-achievable second-order rates for DMC $W$ is denoted as $\mathcal{L}_\varepsilon^*(P_X,W)$.} 
\end{definition}
In view of Theorem~\ref{thm:first-order}, we consider (small) deviations of $R_n$ and~$r_n$ from $C(W)$ and $H_{P_X\times W}(X|Y)$, respectively. Theorem~\ref{thm:first-order} implies that  $C(W)$ and $ H_{P_X\times W}(X|Y)$ are first-order rates that our second-order analysis should be centered at. The difference in the signs in~\eqref{eqn:deviate_rate} and~\eqref{eqn:deviate_abn} is because these are the directions in the ``usual case'' where the   error probability $\varepsilon<1/2$, and we expect that $s$ and $t$ are positive. For completeness, we also treat the case in which $\varepsilon\ge 1/2$.

\red{Unlike most works on second-order rates~\cite{Polyanskiy2010, Tan2014}, fully characterizing $\mathcal{L}_\varepsilon^*(P_X,W)$ requires both an achievability proof and an {\em ensemble} converse, rather than the standard information-theoretic converse such as those involving the use of the meta-converse~\cite{Polyanskiy2010}. This necessity arises because the coding scheme—comprising random codes with constant composition codebooks and guessing-based decoding with abandonment—is {\em predetermined} and {\em  fixed}. Our goal is to precisely determine the performance of this specific scheme without looseness in the analysis. To do so, we  focus on finding the maximal set of second-order rates with a fixed non-vanishing error probability or, equivalently, the smallest asymptotic error probability at fixed second-order rates. This approach parallels  analyses of nearest-neighbor decoding for additive channels~\cite{scarlett2017} and Gaussian codebooks for lossy compression~\cite{zhou2019} where the ensembles of codes and decoding schemes are fixed.}

\begin{figure}[t]
  \centering
   \begin{overpic}[width=.96\columnwidth]{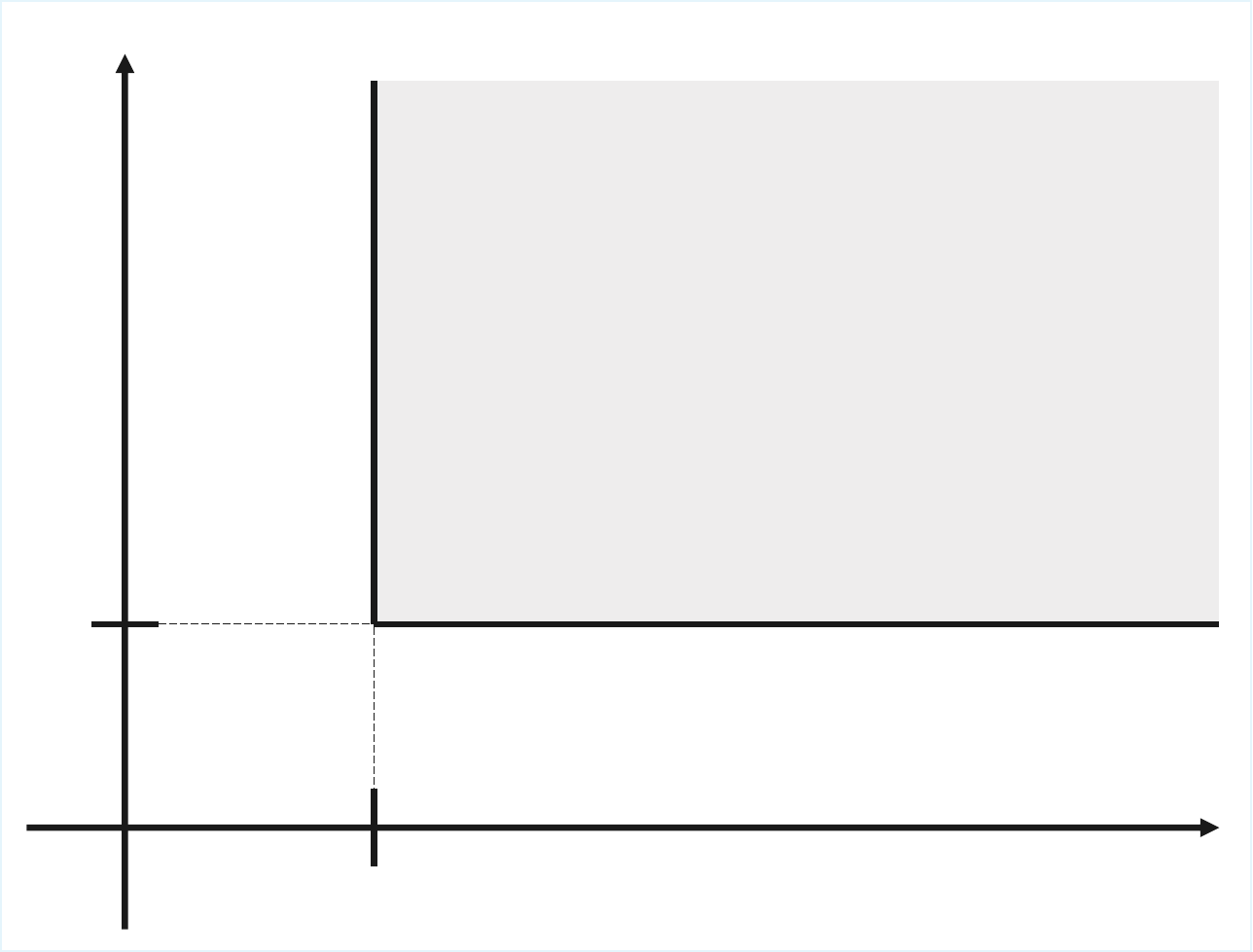}
   \put(96,6){$s$}
   \put(4,70){$t$}
   \put(6,5){$0$}
   \put(23,1){$\sqrt{V_\varepsilon(W)}\mathrm{Q}^{-1}(\varepsilon)$}
   \put(0,19){\rotatebox{90}{$\sqrt{V_\varepsilon(W)}\mathrm{Q}^{-1}(\varepsilon)$}}
   \put(56,46){{\large $\mathcal{L}_\varepsilon^*(P_X,W)$}}
   \end{overpic}
    \caption{Illustration of a second-order region for $\varepsilon \in  (0,1/2)$}
    \label{fig:2nd_order}
\end{figure}

The set of $\varepsilon$-achievable second-order rates, $\mathcal{L}_\varepsilon^*(P_X,W)$, is fully characterized in the following theorem.
\begin{theorem} \label{thm:2nd_order}
Assume that $W$ has positive $\varepsilon$-dispersion. Then, for all $\varepsilon\in[0,1)$,
\begin{align}
    \mathcal{L}_\varepsilon^*(P_X,W) = \left\{ (s,t)\in \mathbb{R}^2 : s\wedge t\ge\sqrt{V_\varepsilon(W)}\mathrm{Q}^{-1}(\varepsilon)\right\}.
\end{align}
\end{theorem}

Theorem~\ref{thm:2nd_order}, which is illustrated schematically in Fig.~\ref{fig:2nd_order}, may be stated alternatively as follows. If the code and abandonment rates are parametrized by $s$ and $t$ as
\begin{align}
    R_n \approx  C(W) - \frac{s }{\sqrt{n}}\;\;\mbox{and}\;\; r_n  \approx  H_{ P_X \times W} (X|Y)+ \frac{t}{\sqrt{n}},
\end{align}
then the ensemble average error probability satisfies
\begin{align}
    \lim_{n\to\infty}
    \varepsilon_n = \left\{  \begin{array}{cc}
        \mathrm{Q}\Big(\frac{s\wedge t}{\sqrt{V_{\min}(W) }} \Big)& \mbox{if }s\wedge t\ge0 \vspace{.3em}\\
         \mathrm{Q}\Big(\frac{s\wedge t}{\sqrt{V_{\max}(W) }} \Big)   & \mbox{if }s\wedge t < 0
    \end{array}\right. .
\end{align}

We provide some intuition for Theorem~\ref{thm:2nd_order} by examining  two  extreme cases. 

If the guessing-based decoder does not incorporate an abandonment step,  $t$ takes the value $+\infty$. Thus, the optimal set of second-order rates is precisely the set of all $(s,t)$ pairs such that $s\ge \sqrt{V_\varepsilon(W)}\mathrm{Q}^{-1}(\varepsilon)$. This reduces to the characterization of the second-order rate for channel coding using the random coding union (RCU) bound \cite{Polyanskiy2010}, which is ensemble tight, and indeed, (second-order) tight over all encoding-decoding  strategies~\cite{Tan2014, Hayashi09, Polyanskiy2010}. 

If, on the other hand, the rate of the code is strictly smaller than $C(W)$, we set $s$ as $+\infty$. In this case, by the definition of the event $\mathcal{A}_1$ in~\eqref{eqn:defA1} and the bounds on the guessing function $G(\bm{x}|\bm{y})$ in terms of $\hat{H}(\bm{x}|\bm{y})$ in~\eqref{eq:G_bounds}, the problem, at a mathematical level, reduces to universal almost-lossless source coding with full side information at the encoder and decoder---{\em universal conditional source coding} in short.
To elaborate, given the channel output $\bm{Y}$, the guessing-based decoder  with abandonment rank-orders  input sequences $\bm{X} \in \mathcal{T}_{P_X^{(n)}}$ up to that of rank $\e^{nr_n}$ according to $G( \cdot | \bm{Y})$ or, almost equivalently, $\hat{H}(\cdot|\bm{Y})$. Here, in the parlance of universal conditional source coding,   the ``source''   $\bm{X} \sim \mathrm{Unif}\big(\mathcal{T}_{P_X^{(n)}} \big)$  is compressed to rate~$r_n$, and the ``side information'' is $\bm{Y}\sim W^n(\cdot|\bm{X})$.  
The error event of this universal conditional source coding setting is exactly $\mathcal{A}_1$.
In this case when $s$ is set as $+\infty$, we also recover the optimal second-order rate of conditional source coding, a result analogous to~\cite{TK14, Nom14}.\footnote{It should be noted that the usual source coding setting~\cite[Ch.~3]{Cov06} involves compressing all source type classes, in contrast to the above setting, which constitutes the sub-problem if compressing a single type class.} 

In summary, the {\em minimum} of the two ``backoff'' terms~$s$ and~$t$ governs the second-order behavior of guessing-based decoding with abandonment and the region  $\mathcal{L}_\varepsilon^*(P_X,W)$ is {\em  rectangular}; this is in contrast to other (known) second-order regions for network information theory which are typically curved~\cite{Tan2014,TK14, scarlettTan2015}. The key reason, as illustrated in the subsequent proof, is that after simplification, the error probability in our setting hinges on the convergence of a sequence of {\em scalar} random variables (see for example, \eqref{eqn:limitA1} and \eqref{eqn:apply_CLT}). In contrast, other network information theory scenarios typically involve analyzing the convergence of sequences of random {\em vectors}.


Finally, we note that the original GRAND paper by Duffy {\em et al.}~\cite{Duffy2019} stated that ``GRANDAB (the GRAND procedure with an abandonment step) results in an error if either the ML decoding is erroneous, ... or if the algorithm abandons guessing before an element of the codebook is identified.'' This assertion is conceptually similar to Theorem \ref{thm:2nd_order}, which can be interpreted as a second-order analogue of this statement. However, it is important to observe that \cite[Prop.~3]{Duffy2019} focuses on issues related to the rank function $G$, rather than the error probability directly, and addresses the error exponent regime. In contrast, Theorem \ref{thm:2nd_order} focuses on the ensemble average error probability in the second-order regime.


\subsection{Proof of the Achievability Part of Theorem~\ref{thm:2nd_order} }
\label{subsec:proof_second_order}
We show that under guessing-based decoding with second-order rates in  $\mathrm{int}(\mathcal{C}_\varepsilon^*(W))$, the ensemble error probability is upper bounded by~$\varepsilon$ as $n\to\infty$. 

For each $n$, the $\e^{nR_n}$ codewords are generated uniformly from a type class $\mathcal{T}_{P_X^{(n)}}$ as described in Section \ref{sec:error_prb_rates}. The sequence $\{P_X^{(n)}\}_{n\in\mathbb{N}}$ is chosen to  converge to a CAID $P_X$ defined in~\eqref{eqn:caid}.
From \cite[Lem~2.1.2]{DemboZeitouni}, $\{P_X^{(n)}\}_{n\in\mathbb{N}}$ can be chosen to converge to $P_X$ in the $\ell_1$ metric at a rate  $O(\frac{1}{n} )$.  
The $\varepsilon$-dispersions $V_\varepsilon(W)$ achieved in these two cases are $V_{\min}(W)$ and $V_{\max}(W)$, respectively. 
Using uniform continuity bounds on the mutual information~\cite{zhang2007} and conditional information variance~\cite{TomTan13}, we have 
$ \big|I(P_X^{(n)},W)-I(P_X,W) \big|= O\big( \frac{\log n}{n} \big)$ and $ \big|V(P_X^{(n)},W)-V(P_X,W) \big|= O\big( \frac{1}{n} \big)$. 
This implies that the first- and second-order terms will be unaffected by the approximation of distributions by types. Thus, in the following, for the sake of brevity, we write $P_{\bm{X}}$ as the uniform distribution over the type class $\mathcal{T}_{P_X^{(n)}}$ without referencing the blocklength~$n$, and ignoring approximations of distributions by types. 

For each $r\ge0$, define the set
\begin{align}
\mathcal{S}_{r} := \big\{  (\bm{x} , \bm{y} ): G(\bm{x}|\bm{y})\le\e^{nr} \big\}.   \label{eqn:def_Sr}
\end{align}
By the law of total probability, starting from 
$\varepsilon_n = \Prb \left[ \mathcal{E}_1 \cup \mathcal{A}_1\right]$ \red{(please refer to  the definitions of $\mathcal{A}_1$ and $\mathcal{E}_1$ in~\eqref{eqn:def_E1} and~\eqref{eqn:defA1} respectively),} we have 
$\varepsilon_n = (\clubsuit)+ (\spadesuit)$
where 
\begin{align}
  (\clubsuit)&:=  \sum_{ (\bm{x}, \bm{y})  \in \mathcal{S}_{r_n} } P_{\bm{X}}(\bm{x})W^n(\bm{y} | \bm{x})\varphi(\bm{x},\bm{y}) \quad\mbox{and} \label{eqn:decomp_Tr1} \\
  (\spadesuit)&:= \sum_{ (\bm{x}, \bm{y})  \in \mathcal{S}_{r_n}^{\comp} } P_{\bm{X}}(\bm{x})W^n(\bm{y} | \bm{x})\varphi(\bm{x},\bm{y})  , \label{eqn:decomp_Tr2}
\end{align}
and 
\begin{equation}
\varphi(\bm{x},\bm{y}):= \Prb \left[ \mathcal{E}_1 \cup \mathcal{A}_1 \mid \bm{X}(1) =\bm{x}, \bm{Y}=\bm{y} \right]  . \label{eqn:varphi_def}    
\end{equation}
For all $(\bm{x}, \bm{y})  \in \mathcal{S}_{r_n}^{\comp} $, it is clear from the definition of the event $\mathcal{A}_1$ that $\varphi(\bm{x},\bm{y})=1$. Thus,  $(\spadesuit)$ in~\eqref{eqn:decomp_Tr2} is simply $\Prb[\mathcal{A}_1]$, which in view of \eqref{eqn:PA1_UB} and \eqref{eqn:PA1_LB} and by standard second-order analysis~\cite{wang2011} (the central limit theorem), 
satisfies
\begin{align}
    \lim_{n\to\infty}\Prb[\mathcal{A}_1]=\mathrm{Q}\bigg( \frac{t}{\sqrt{V_\varepsilon(W)}}\bigg). \label{eqn:limitA1}
\end{align}
We note, by the choice of $P_X$ in \eqref{eqn:caid}, that $V_\varepsilon(W)$ is    $V_{\min}(W)$ or $V_{\max}(W)$ if $s\wedge t\ge0$ or $s\wedge t<0$, respectively.

We now focus our attention on $(\clubsuit)$. Note from the definition of  $\mathcal{S}_{r_n}$ that   $\varphi(\bm{x},\bm{y})=\Prb \left[ \mathcal{E}_1  \mid \bm{X}(1) =\bm{x}, \bm{Y}=\bm{y} \right]$. Using the independence of the codewords, we deduce that
\begin{align}
    &\varphi(\bm{x},\bm{y}) \nonumber\\*
    &=1\!-\! \prod_{w\ne 1}\Prb \Big[  G(\bm{X}(w) | \bm{Y}) \!>\! G(\bm{X}(1)|\bm{Y}) \,\big| \, \bm{X}(1)\!=\!\bm{x},\bm{Y}\!=\!\bm{y}  \Big] \nonumber\\
    &=1- \prod_{w\ne 1}\big(1-\Prb \left[  G(\bm{X}(w) | \bm{y}) \le G(\bm{x}|\bm{y})   \right] \big)\nonumber\\
    &=1-\big(1-\Psi(\bm{x}, \bm{y})\big)^{M_n-1}, \label{eqn:PE1}
\end{align}
where 
\begin{equation}
\Psi(\bm{x}, \bm{y}) :=\Prb\left[G(\overline{\bm{X}}|\bm{y})\le G(\bm{x}|\bm{y})\right]    \label{eqn:def_Psi}
\end{equation}
and $\overline{\bm{X}} \sim P_{\bm{X}}$ is independent of $(\bm{X},\bm{Y})$.
  Intuitively, for any given $(\bm{x}, \bm{y})$, the quantity $\Psi(\bm{x}, \bm{y})$ represents the probability that a non-transmitted codeword $\overline{\bm{X}}$ has a smaller rank than that of the true transmitted codeword $\bm{x}$ given the channel output $\bm{y}$. 
Before we proceed, we present the following useful lemma whose proof can be found in Appendix~\ref{app:prf_bound_psi}.
\begin{lemma}[Bounds on $\Psi(\bm{x},\bm{y} )$]\label{lem:bound_psi}
For every $n \in \mathbb{N}$, 
\begin{align}
   (n\!+\!1)^{-2|\mathcal{X}||\mathcal{Y}|}\e^{-n\hat{I} (\bm{x}\wedge\bm{y}) } \!\le\!\Psi( \bm{x},\bm{y} )\!\le\!(n\!+\!1)^{3|\mathcal{X}||\mathcal{Y}|}\e^{-n\hat{I} (\bm{x}\wedge\bm{y}) }.
\end{align}
\end{lemma}
To facilitate the subsequent analysis, define the set 
\begin{align}
    \mathcal{F}_n:=\left\{ (\bm{x},\bm{y}): \Psi(\bm{x}, \bm{y}) \le  \frac{1}{n(M_n-1)}\right\}. \label{eqn:defU_n}
\end{align}
Observe that for all $ (\bm{x},\bm{y})\in\mathcal{F}_n$, 
\begin{align}
     &1-\big(1-\Psi(\bm{x}, \bm{y})\big)^{M_n-1}\nonumber\\
     &\le 1-\bigg(1- \frac{1}{n(M_n-1)} \bigg)^{M_n-1} \nonumber\\
     &\le 1-\bigg[\exp\bigg( -\frac{\frac{1}{n(M_n-1)}}{1-\frac{1}{n(M_n-1)}} \bigg)  \bigg]^{M_n-1} \label{eqn:useineq_e}\\
    &=1- \exp \left( - \frac{1}{n-(M_n-1)^{-1}} \right) =:\eta_n\to 0\label{eqn:def_delta},
\end{align}
where \eqref{eqn:useineq_e} follows from the   inequality $1-t \ge\e^{  -\frac{t}{1-t}}$ for all $t<1$.  By Taylor's theorem,  $\eta_n$ is of the order $O(\frac{1}{n})$. 

By splitting $(\clubsuit)$ into two parts ($(\bm{x}, \bm{y})  \in \mathcal{S}_{r_n}\cap\mathcal{F}_n^{\comp}$ and $(\bm{x}, \bm{y})  \in \mathcal{S}_{r_n}\cap\mathcal{F}_n$) and using \eqref{eqn:def_delta}, 
\begin{align}
     (\clubsuit) 
     &\le \sum_{ (\bm{x}, \bm{y})  \in \mathcal{S}_{r_n}\cap\mathcal{F}_n^{\comp} } P_{\bm{X}}(\bm{x})W^n(\bm{y} | \bm{x}) + \eta_n. \label{eqn:two_ubs}
\end{align}
It remains to upper bound the first term in \eqref{eqn:two_ubs}, which can be expressed as 
\begin{align}
 &\Prb\left[ (\bm{X}(1),\bm{Y})\in \mathcal{S}_{r_n} \cap \mathcal{F}_n^{\comp}   \right] \nonumber\\
    &=\Prb \left[  \Big\{G(\bm{X}|\bm{Y})\le\e^{nr_n}\Big\}\cap\Big\{\Psi(\bm{X},\bm{Y})>\frac{1}{n(M_n-1)}\Big\}\right]  \nonumber\\
    &\le \Prb \bigg[  \left\{   \e^{n \hat{H}(\bm{X}|\bm{Y})  }\le\e^{nr_n}\right\} \cap \nonumber\\*
    &\qquad\quad \Big\{(n+1)^{ 3 |\mathcal{X}| |\mathcal{Y}| } \e^{-n \hat{I}(\bm{X} \wedge \bm{Y})}>\frac{1}{n(M_n-1)}  \Big\}\bigg] .\label{eqn:twoevents}
\end{align}
In \eqref{eqn:twoevents}, we used the lower bound on $G( \bm{x}|\bm{y})$ in terms of the empirical conditional entropy $\hat{H}( \bm{x}|\bm{y})$  in \eqref{eq:G_bounds} and the  upper bound on $\Psi(\bm{x},\bm{y})$ stated in   Lemma~\ref{lem:bound_psi}.

Using the bounds on $r_n$ and $R_n = \frac{1}{n}\log M_n$ in terms of $t$ and $s$ in~\eqref{eqn:deviate_rate} and~\eqref{eqn:deviate_abn} respectively, we observe that the two events in~\eqref{eqn:twoevents} can be  expressed  as 
\begin{align}
    \hat{I}(\bm{X}\wedge \bm{Y})&\ge I(P_X,W) - \frac{t}{\sqrt{n}} + o\Big( \frac{1}{\sqrt{n}}\Big) \quad \mbox{and} \label{eqn:event1}\\
    \hat{I}(\bm{X}\wedge \bm{Y})&< I(P_X,W) - \frac{s}{\sqrt{n}} + o\Big( \frac{1}{\sqrt{n}}\Big) .\label{eqn:event2}
\end{align}
The intersection of the two events is empty  for large enough~$n$ if $s>t$. On the other hand, if $s\le t$, the intersection is not empty. \red{In fact, for any $\eta>0$ and for large enough~$n$, we have 
\begin{align}
    &\Prb\left[ (\bm{X}(1),\bm{Y})\in \mathcal{S}_{r_n} \cap \mathcal{F}_n^{\comp}   \right]\nonumber\\*
    &\le \Prb\left[ \hat{I}(\bm{X}\wedge \bm{Y})\!-\! I(P_X,W)\in \Big[ -\frac{t+\eta}{\sqrt{n}}, -\frac{s-\eta}{\sqrt{n}} \Big)   \right] \nonumber\\
    &=\mathrm{Q}\bigg( \frac{s-\eta}{\sqrt{V_\varepsilon(W)}}\bigg)- \mathrm{Q}\bigg( \frac{t+\eta}{\sqrt{V_\varepsilon(W)}}\bigg) + o(1),\label{eqn:apply_CLT}
\end{align}}
where~\eqref{eqn:apply_CLT} follows from applying the central limit theorem on the empirical mutual information~\cite{wang2011}. 
Uniting \eqref{eqn:limitA1}, \eqref{eqn:two_ubs}, and~\eqref{eqn:apply_CLT},  taking $n\to\infty$, \red{then taking $\eta\to0$} yields
\begin{align}
    \limsup_{n\to\infty}\varepsilon_n\le\left\{ \begin{array}{cc}
         \mathrm{Q}\Big( \frac{s}{\sqrt{V_\varepsilon(W)}}\Big)&\mbox{if }s\le t \vspace{.3em}  \\
         \mathrm{Q}\Big( \frac{t}{\sqrt{V_\varepsilon(W)}}\Big)&\mbox{if } s>t
    \end{array}  \right. .
\end{align}
which is the desired upper bound on the asymptotic error probability. 

\begin{remark}
   Unlike the proof of the first-order result (and the proof involving the error exponent in Section \ref{sec:error_exp}, but not the strong converse exponent in Section \ref{sec:strong_conv_exp}), applying the union bound on $\Prb \left[ \mathcal{E}_1 \cup \mathcal{A}_1\right]$ as in~\eqref{eq:error_union_max_bounds}   results  in suboptimality in the   asymptotic ensemble error probability. Thus, in the above proof, a more intricate computation is used, taking into account {\em both} the decoding error event $\mathcal{E}_1$ and the abandonment event $\mathcal{A}_1$ through the introduction of the set $\mathcal{F}_n$ in~\eqref{eqn:defU_n}.
\end{remark}

\subsection{Proof of the Ensemble Converse Part of Theorem~\ref{thm:2nd_order} }
We now prove the ensemble converse part. Since all the derivations in the achievability part are equalities up to and including \eqref{eqn:PE1}, the derivations up that point hold {\em mutatis mutandis}. The point of deviation is in the definition of the set $\mathcal{F}_n$ in~\eqref{eqn:defU_n}. Inspired by \cite{scarlett2017}, we now consider the set 
\begin{align}
    \mathcal{G}_n:=\left\{ (\bm{x},\bm{y}): \Psi(\bm{x}, \bm{y}) \ge  \frac{n}{  M_n-1 }\right\}.
\end{align}
Observe that for all $ (\bm{x},\bm{y})\in\mathcal{G}_n$, 
\begin{align}
     &1-\big(1-\Psi(\bm{x}, \bm{y})\big)^{M_n-1} \ge 1-\Big(1- \frac{ n}{ M_n-1 } \Big)^{M_n-1} \\*
     &\ge 1-\big(\e^{- \frac{n}{M_n-1}}\big)^{M_n-1} \ge 1- \e^{-n},\label{eqn:oneminus_lb}
\end{align}
where the second inequality follows from  $1-t\le \e^{-t}$ for all~$t$.

Using the bound in~\eqref{eqn:oneminus_lb}, the term $(\clubsuit)$ in~\eqref{eqn:decomp_Tr1} can be further lower bounded as follows:
\begin{align}
 (\clubsuit)
     &\ge(1-\e^{-n}) \Prb \left[ (\bm{X}(1),\bm{Y})\in \mathcal{S}_{r_n}\cap\mathcal{G}_n  \right].
\end{align}
We observe that the probability above can be expressed as 
\begin{align}
 &\Prb \left[ (\bm{X}(1),\bm{Y})\in \mathcal{S}_{r_n}\cap\mathcal{G}_n  \right]\nonumber\\*
     &= \Prb \left[  \Big\{G(\bm{X}|\bm{Y})\!\le\!\e^{nr_n}\Big\}\cap \Big\{\Psi(\bm{X},\bm{Y})\!\ge\!\frac{n}{ M_n\!-\! 1 }\Big\}\right] \nonumber\\
    &\ge \Prb \bigg[  \left\{  (n+1)^{ |\mathcal{X}||\mathcal{Y}|} \e^{n \hat{H}(\bm{X}|\bm{Y})  }\le\e^{nr_n}\right\}\cap \nonumber\\*
    &\qquad\Big\{(n+1)^{  -2|\mathcal{X}| |\mathcal{Y}| } \e^{-n \hat{I}(\bm{X} \wedge \bm{Y})} \ge \frac{n}{ M_n-1 }  \Big\}\bigg] ,\label{eqn:twoevents_ens_conv}
\end{align}
where \eqref{eqn:twoevents_ens_conv} uses the upper bound on $G(\bm{x}|\bm{y})$ in~\eqref{eq:G_bounds} and the lower bound of $\Psi(\bm{x},\bm{y})$ in Lemma~\ref{lem:bound_psi}.

Similarly to the achievability part, we observe that the two events in~\eqref{eqn:twoevents_ens_conv} can be expressed as in~\eqref{eqn:event1} and~\eqref{eqn:event2}. An application of the central limit theorem then yields
\begin{align}
    (\clubsuit)\!\ge  \!  (1\!-\!\e^{-n})\bigg[  \mathrm{Q}\bigg( \frac{s}{\sqrt{V_\varepsilon(W)}}\bigg)\!-\! \mathrm{Q}\bigg( \frac{t}{\sqrt{V_\varepsilon(W)}}\bigg) \! + \! o(1)\bigg].\label{eqn:apply_CLT2}
\end{align}
Combining this lower bound with the computation of $(\spadesuit)$ in~\eqref{eqn:limitA1} yields the desired ensemble converse, matching the achievability part. 

\section{Ensemble-Tight Error Exponents} \label{sec:error_exp}
In this section, we consider the regime in which  the code rate $R_n$ and abandonment rate $r_n$ are fixed while the ensemble error probability $\varepsilon_n$ decays exponentially fast with the blocklength. 
Since these rates do not vary with the blocklength in this setting, we write them as~$R$ and~$r$ respectively. 
Our objective here is to derive the error exponent as a function of $R$ and $r$.
To this end, for fixed $R$, $r$, $P_X$ and $W$, we define\footnote{The exponents clearly depend on the DMC $W$, but this dependence is suppressed for the sake of brevity.} 
\begin{align}
    E_\mathrm{r}(R,P_X) & := \min_{V} D(V\|W | P_X)  + \left| I(P_X,V) - R \right|^{+}, \\
    E_{\mathrm{a}}(r,P_X) & := \min_{V: I(P_X,V) \leq H(P_X) - r} D(V\|W | P_X).
\end{align}
Note that $E_\mathrm{r}(R,P_X)$ is the random coding error exponent for constant composition codes~\cite{Csi97}, where $P_X$ is the limit of the $n$-types $P_X^{(n)}$ (note that in this section, $P_X$  is not necessarily a CAID).  As we shall see in the proof of Theorem~\ref{thm:ee}, $E_{\mathrm{a}}(r,P_X)$ turns out to be the abandonment exponent. Recall Haroutunian's sphere packing exponent~\cite{haroutunian68}
\begin{align}
    E_{\mathrm{sp}}(R,P_X)   := \min_{V: I(P_X,V) \leq R} D(V\|W | P_X);
\end{align}
also see Blahut~\cite{blahut74}. We remark that $E_{\mathrm{a}}(r, P_X)$ is equal to  $E_{\mathrm{sp}}(R,P_X)$ at rate $R = H(P_X) - r$.

In analogy to Definition~\ref{def:2nd_order}, we define the ensemble error exponent as follows.
\begin{definition} \label{def:err_exp}
Fix a limiting distribution $P_X$ (such that $P_X^{(n)}\to P_X$). The exponent $E\in\mathbb{R}_+$ is  {\em achievable} at rate pair $(R,r)$ if the sequence of $(n,R,r, P_{X}^{(n)})$-codes with \red{random-coding} error probabilities $\{\varepsilon_n\}_{n\in\mathbb{N}}$ satisfies  
\begin{align}
    \liminf_{n\to\infty}\frac{1}{n}\log\frac{1}{\varepsilon_n}&\ge E. \label{eqn:ach_error_exp}
\end{align}
The supremum of all achievable exponents for a given rate pair $(R,r)$, limiting input distribution $P_X$ and DMC $W$ is denoted as $E^*(R,r,P_X)$.  Maximizing $E^*(R,r,P_X)$ over $P_X$ yields the {\em maximum ensemble error exponent} $E^*(R,r)$.
\end{definition}
We emphasize that similarly to Section~\ref{sec:second_order},  {\em  ensemble-tight results} are sought here as well. Each codebook $\mathsf{C}_n$  is chosen randomly. The  codewords in  $\mathsf{C}_n$  are independent and drawn uniformly from the type class $\mathcal{T}_{P_X^{(n)}}$. The decoding procedure is also fixed according to the guessing-based decoding with abandonment approach as per the discussion in Section~\ref{sec:error_prb_rates}. 
\begin{theorem}\label{thm:ee}
    For any $P_X \in \mathcal{P}(\mathcal{X})$ and $(R,r) \in \mathbb{R}^2_{+}$,
\begin{align}
   E^*(R,r,P_X)   = \min  \big\{ E_{\mathrm{r}}(R,P_X),  E_{\mathrm{a}}(r,P_X) \big\}.
    \label{eqn:E_star}
\end{align}
Moreover, $E^*(R,r) = \max_{P_X}E^*(R,r,P_X ) $.
\end{theorem}
Using the known properties of the random coding exponent $E_{\mathrm{r}}$ and the sphere packing exponent $E_{\mathrm{sp}}$, we deduce that the maximum ensemble exponent $E^*(R,r)>0$ if and only if $R < C(W)$ {\em and}   $r>H_{P_X\times W}(X|Y)$, where $P_X$ is a CAID. 

For any fixed $P_X$, let $R_{\mathrm{cr}}(P_X,W)$ be the critical rate above which we have $E_{\mathrm{sp}}(R,P_X) = E_{\mathrm{r}}(R,P_X) $ \cite{Csi97}.
In the regime  $R \geq R_{\mathrm{cr}}(P_X,W)$ and $H(P_X)-r \geq R_{\mathrm{cr}}(P_X,W)$, the ensemble exponent in Theorem \ref{thm:ee} simplifies  to 
\begin{equation}
\label{eqn:E_star_above_Rcr}
E^*(R,r,P_X) =  E_{\mathrm{r}}(\max\{R , H(P_X)-r\},P_X).
\end{equation}
Here we can see two distinct effects that influence the ensemble error exponent in relation to the code rate $R$ and the abandonment rate $r$. If the abandonment rate is sufficiently large, more precisely $r\ge H(P_X)-R$, then $E^*(R,r,P_X)$ in~\eqref{eqn:E_star_above_Rcr} reduces to the random coding exponent with input distribution $P_X$ at code rate $R$. On the other hand, if $r$ is sufficiently small, specifically, $r < H(P_X)-R$, then $E^*(R,r,P_X)$ reduces to $\min_V\{ D(V\|W|P_X) + | r- H_{P_X\times V}(X|Y)|^+\}$. This coincides with the error exponent of conditional source coding, $\bm{X} \sim \mathrm{Unif}\big(\mathcal{T}_{P_X} \big)$ being the source and $\bm{Y} \sim W^n(\cdot | \bm{X})$ the side information, a result that can be deduced from~\cite{csiszar82,   oohamaHan94, draper07}. This interpretation is similar to that for our second-order result. 

Finally, by setting $r=H(P_X)-R$ in~\eqref{eqn:E_star}, we recover the usual random coding error exponent $E_{\mathrm{r}}(R,P_X)$. Moreover, since codeword-centric decoding requires $\e^{nR}$ guesses while noise-centric decoding \cite{Duffy2019} requires $\e^{nr}$ guesses, the latter is more efficient wherever $r < R$. This boils down to $R > \frac{1}{2}H(P_X)$ under the condition $r=H(P_X)-R$.

\begin{figure}
    \centering
    \includegraphics[width=0.95\linewidth]{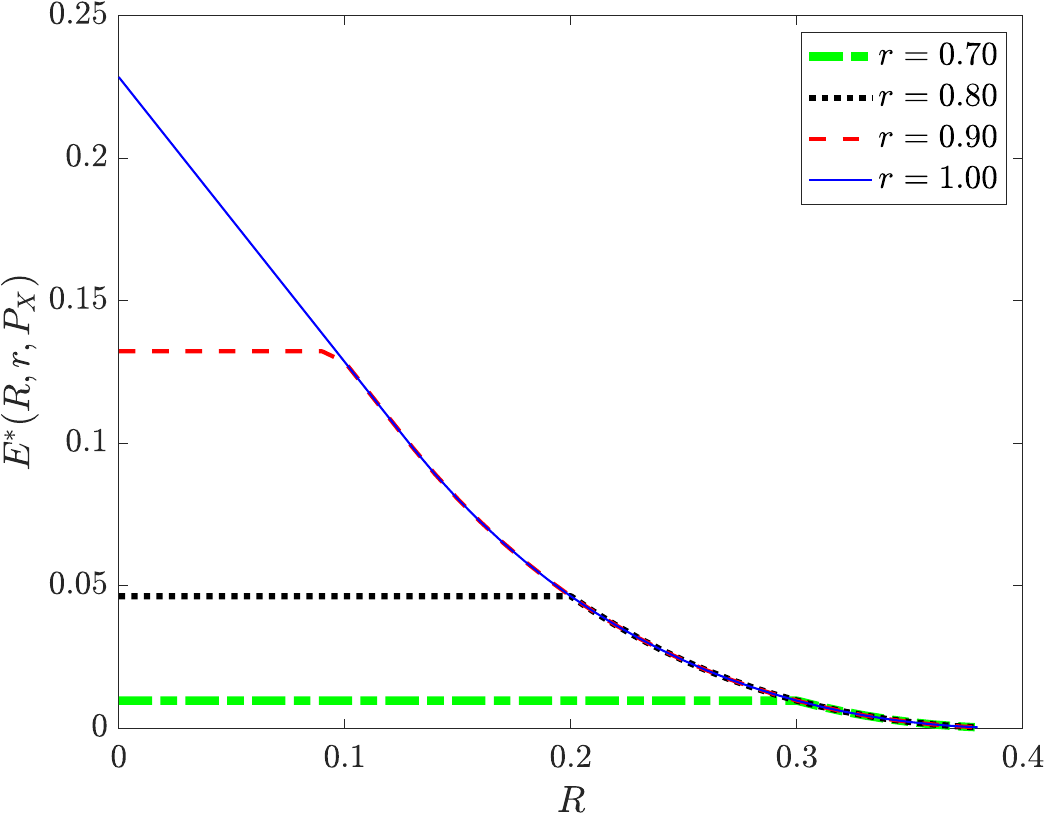}
    \caption{An illustration of $E^*(R,r,P_X)$}
    \label{fig:Estar}
\end{figure}

\begin{example} \red{{\em The error exponent $E^*(R,r,P_X)$ is illustrated in Fig.~\ref{fig:Estar} for the binary asymmetric channel
\begin{align}
    W= \begin{bmatrix}
    0.8   &      0.2\\
0.1      &   0.9
    \end{bmatrix} \label{eqn:bac}
\end{align}
with input distribution $P_X = \mathrm{Ber}(1/2)$. Note that when the abandonment rate $r$ attains its maximum value of $1$ (blue solid line), the exponent simplifies to the  random coding error exponent $E_{\mathrm{r}}(R, P_X)$. For a fixed $r < 1$, the exponent initially remains   flat at small values of $R$, before coinciding with $E_{\mathrm{r}}(R, P_X)$ as $R$ increases. Conversely, for a fixed code rate $R$, decreasing $r$ results in a smaller  exponent, reflecting the fact that decoding is abandoned after fewer guesses.}}
\end{example}

\subsection{Proof of the Achievability Part of Theorem~\ref{thm:ee} }
\label{subsec:achievability_ee}
\red{Fix an} input distribution $P_X$ and a sequence of $n$-types $P_X^{(n)}$ that converges to $P_X$. Since the approximation of distribution by types does not affect the exponents (by continuity), for brevity, we   write $P_X^{(n)}$ as $P_X$ in this and subsequent proofs.

Recall from \eqref{eq:error_union_max_bounds} that the ensemble average error probability is bounded above as $ \varepsilon_n   \leq   \Prb \left[ \mathcal{E}_1\right] + \Prb \left[ \mathcal{A}_1\right]$.
Next, we treat the incorrect decoding probability $\Prb \left[ \mathcal{E}_1\right]$ and the abandonment probability $\Prb \left[ \mathcal{A}_1\right]$ separately. Starting with $\Prb \left[ \mathcal{E}_1\right]$, bounding this is almost identical to bounding the error probability of the MMI decoder. For completeness, we present a proof which is adapted to our setting. We start with
\begin{align}
&\Prb \left[ \mathcal{E}_1\right]  = \Prb \bigg[ \bigcup_{w \neq 1} \left\{ G(\bm{X}(w) | \bm{Y}) \le G(\bm{X}(1)|\bm{Y}) \right\} \bigg] \\ 
& = \E \Bigg[  \Prb \bigg[ \bigcup_{w \neq 1} \left\{ G(\bm{X}(w) | \bm{Y}) \le G(\bm{X}(1)|\bm{Y}) \right\} \Big|\bm{X}(1),\bm{Y} \bigg] \Bigg] \\
& \leq\!  \E \Bigg[ \! \min\bigg\{\! 1 , \!\sum_{w \ne 1}\! \Prb \big[  G(\bm{X}(w) | \bm{Y})\! \le\! G(\bm{X}(1)|\bm{Y}) \big|\bm{X}(1),  \bm{Y}   \big] \!\bigg\} \!\Bigg] \\
& =  \E\Big[ \min \big\{1 , (\e^{nR}-1)\Prb \left[ G(\overline{\bm{X}} | \bm{Y})\le G(\bm{X}|\bm{Y}) \big| \bm{X}, \bm{Y}\right]  \big\}  \Big]  \\
& =  \E\Big[ \min \big\{1 , (\e^{nR}-1)\Psi(\bm{X},\bm{Y})  \big\}  \Big] ,
\label{eq:RCU_UB}
\end{align}
where, as in Section \ref{subsec:proof_second_order},  $\overline{\bm{X}} \sim P_{\bm{X}}$ is independent of $(\bm{X},\bm{Y})$, distributed identically as $\bm{X}$, and $\Psi(\bm{X},\bm{Y}) $ is as defined in \eqref{eqn:def_Psi}.
Note that \eqref{eq:RCU_UB} is the so-called random coding union bound (RCU) \cite{Polyanskiy2010}, or more precisely, a mismatched version of it \cite{scarlett_mismatched}.

According to Lemma~\ref{lem:bound_psi}, we can upper bound $\Psi(\bm{x},\bm{y})$  by $(n+1)^{3 |\mathcal{X}| |\mathcal{Y}|} \e^{ -n\hat{I} (\bm{x}\wedge\bm{y})}$. Plugging this upper bound back into~\eqref{eq:RCU_UB}, while conditioning on $\bm{X} = \bm{x}$,
we obtain
\begin{align}
&\Prb \left[ \mathcal{E}_1 | \bm{X} = \bm{x} \right] \nonumber\\*
& \leq  \sum_{\bm{y} \in \mathcal{Y}^n} W^n(\bm{y}|\bm{x})\min \left\{ 1, (n+1)^{3|\mathcal{X}||\mathcal{Y}|} \e^{nR-n\hat{I}(\bm{x} \wedge \bm{y})} \right\} \\
 & \le \sum_{V \in \mathcal{V}_n(\mathcal{Y};P_X)} \e^{-n D(V\|W|P_X)} \nonumber\\* 
&\qquad\qquad\times \min \left\{ 1, (n+1)^{3|\mathcal{X}||\mathcal{Y}|}  \e^{nR-nI(P_X, V )} \right\} \label{eqn:prob_Vshell} \\
 & \leq  (n\!+\!1)^{3|\mathcal{X}||\mathcal{Y}|}\!\!\sum_{V \in \mathcal{V}_n(\mathcal{Y};P_X)}\! \e^{-n D(V\|W|P_X)}  \e^{-n|I(P_X, V ) - R |^+}  \\
 & \leq (n\!+\!1)^{4|\mathcal{X}||\mathcal{Y}|}  \max_{V \in \mathcal{P}(\mathcal{Y}|\mathcal{X})}  \e^{-n ( D(V\|W | P_X) + |I(P_X,V)- R|^{+} )} \\
 & = \e^{- n (E_{\mathrm{r}}(R,P_X) - 4\delta_n)} , \label{eqn:PE1_ach}
\end{align}
where~\eqref{eqn:prob_Vshell} follows from splitting the sum over $\bm{y}$ into $V$-shells and noting that $W^n( \mathcal{T}_V( \bm{x})|\bm{x})\le\e^{-nD(V\|W|P_X)} $ \cite[Lemma~2.6]{Csi97} and in \eqref{eqn:PE1_ach}, $\delta_n := \frac{|\mathcal{X}||\mathcal{Y}|}{n}\log(n+1 )$.
Since the bound only depends on $\bm{x}$ through its type $P_X$, and $\bm{X}$ is supported on the type class $\mathcal{T}_{P_X}$, then the bound is also valid for $\Prb \left[ \mathcal{E}_1 \right]$.

Now we proceed to bound the probability of abandonment 
\begin{equation}
    \Prb \left[ \mathcal{A}_1 \right] = \Prb \left[ G(\bm{X}(1)|\bm{Y}) > \e^{nr}  \right]  = \Prb \left[ G(\bm{X}|\bm{Y}) > \e^{nr}  \right].
\end{equation}
From the upper bound in~\eqref{eq:G_bounds}, we know that $G(\bm{x}|\bm{y}) > \e^{nr} $ implies that $\hat{H}(\bm{x}|\bm{y}) \geq r - \delta_n$.
We proceed to bound $\Prb \left[ \mathcal{A}_1 \right]$ while conditioning on $\bm{X} = \bm{x}$ as follows
\begin{align}
  &\Prb \left[ \mathcal{A}_1 | \bm{X} = \bm{x} \right] \nonumber\\
  & \leq \Prb \big[ \hat{H}(\bm{x}|\bm{Y}) \geq r  - \delta_n \big] \\
  & =\sum_{\bm{y} \in \mathcal{Y}^n} W^n(\bm{y}|\bm{x}) \ind\big[ \hat{I}(\bm{x}\wedge\bm{y}) \leq \hat{H}(\bm{x}) - r + \delta_n   \big] \\
  & \le \sum_{V \in \mathcal{V}_n(\mathcal{Y};P_X )}  \e^{-n D(V\|W | P_X)}  \nonumber\\*
  &\qquad\qquad\times\ind \left[  I(P_X,V) \le  H(P_X) - r+\delta_n  \right] \\
& \leq (n+1)^{|\mathcal{X}||\mathcal{Y}|}\max_{V : I(P_X,V) \leq H(P_X) - r + \delta_n } \e^{- n D(V\|W |P_X) } \\
& = \e^{- n E_{\mathrm{a}}(r - \delta_n ,P_X) + n \delta_n}.
\end{align}
Yet again, the bound depends on $\bm{x}$ only through its type $P_X$, and is therefore valid for $\Prb \left[ \mathcal{A}_1 \right]$.

From the two bounds, we proceed to upper bound $\varepsilon_n$ as 
\begin{align}
\varepsilon_n & \leq \e^{- n (E_{\mathrm{r}}(R,P_X) - 4 \delta_n)} + \e^{- n (E_{\mathrm{a}}(r - \delta_n,P_X) -\delta_n )}  \\
& \leq 2\e^{4n \delta_n}\max \Big\{\e^{- n E_{\mathrm{r}}(R,P_X)}, \e^{- n E_{\mathrm{a}}(r - \delta_n,P_X)} \Big\} \\
& = \e^{- n ( \min\{ E_{\mathrm{r}}(R,P_X), E_{\mathrm{a}}(r - \delta_n,P_X) \} - 4\delta_n - \frac{1}{n} \log 2)}.
\end{align}
Since $\delta_n \to 0$ as $n \to \infty$, and by continuity, it follows that 
\begin{equation}
 \!\liminf_{n \to \infty} \frac{1}{n}\log\frac{1}{\varepsilon_n } \! \geq \!
 \min\{ E_{\mathrm{r}}(R,P_X,W), E_{\mathrm{a}}(r,P_X,W) \}  \label{eqn:achi_err_exp} 
\end{equation}
which concludes the achievability part.
\subsection{Proof of the Ensemble Converse Part of Theorem~\ref{thm:ee} }
We now  show that \eqref{eqn:achi_err_exp}  is asymptotically tight. 
For this, we start from  $\varepsilon_n \geq \max \left\{  \Prb \left[ \mathcal{E}_1 \right], \Prb \left[ \mathcal{A}_1 \right]   \right\}$.
Mirroring the achievability proof, we lower bound $\Prb \left[ \mathcal{E}_1 \right]$ and $\Prb \left[ \mathcal{A}_1 \right] $ individually. 

Starting with $\Prb \left[ \mathcal{E}_1 \right]$, we first observe that the RCU-like upper bound in \eqref{eq:RCU_UB}, which is a truncated union bound, is in fact tight up to a factor of $\frac{1}{2}$ for independent events~\cite[Lemma~A.2]{Shulman}; see also Scarlett {\em et al.}~\cite{scarlett_mismatched}. Here, the mutual independence of the codewords $\bm{X}(w)$, $w\in [M_n]$, in the random codebook ensemble is essential. Thus, we have 
\begin{align}
&\Prb \left[ \mathcal{E}_1 \right] \geq \frac{1}{2}  \E\Big[ \min \big\{1 , (\e^{nR}-1)\Psi(\bm{X},\bm{Y})  \big\}  \Big] .\label{eqn:rcu_lb}
\end{align}
Next, we recall from  Lemma~\ref{lem:bound_psi} that $\Psi(\bm{x},\bm{y})$ is lower bounded by $(n\!+\!1)^{-2|\mathcal{X}||\mathcal{Y}|}\e^{-n\hat{I} (\bm{x}\wedge\bm{y}) } $. Plugging this into \eqref{eqn:rcu_lb}, while conditioning on $\bm{X}=\bm{x}$, we obtain
\begin{align}
    &\Prb \left[ \mathcal{E}_1|\bm{X}=\bm{x} \right] \nonumber\\*  
    &\ge\frac{1}{2} \sum_{\bm{y} \in \mathcal{Y}^n} W^n(\bm{y}|\bm{x})\min \left\{ 1, (n\!+\!1)^{ -2 |\mathcal{X}||\mathcal{Y}|} \e^{nR-n\hat{I}(\bm{x} \wedge \bm{y})} \right\} \\
    &\ge\frac{1}{2} \sum_{V \in \mathcal{V}_n(\mathcal{Y};P_X)} (n+1)^{-|\mathcal{X}||\mathcal{Y}|}  \e^{-n D(V\|W|P_X)} \nonumber\\* 
&\qquad\qquad\times \min \left\{ 1, (n+1)^{-2|\mathcal{X}||\mathcal{Y}|}  \e^{nR-nI(P_X, V )} \right\} \\
& \ge \frac{1}{2}\e^{- n E_{\mathrm{r}}(R,P_X) - 3n \delta_n} , \label{eqn:PE1_conv}
\end{align}
where \eqref{eqn:PE1_conv} follows from similar steps as those leading to \eqref{eqn:PE1_ach}.

Next, we proceed to lower bound $\Prb \left[ \mathcal{A}_1 \right] $.
Conditioning on $\bm{X} = \bm{x}$, and recalling that $m_n = \e^{nr}$ and $G(\bm{x}|\bm{y})$ both take positive integer values, we observe that
\begin{align}
\Prb \big[ \mathcal{A}_1 | \bm{X} = \bm{x} \big] & =  \Prb \big[ G(\bm{x}|\bm{Y}) > \e^{nr} \big] \\
& = \Prb \big[ G(\bm{x}|\bm{Y}) \geq \e^{nr} + 1  \big] \\
& \geq \Prb \big[ G(\bm{x}|\bm{Y}) \geq 2\e^{nr} \big].
\end{align}
From the lower bound in~\eqref{eq:G_bounds}, we know that $G(\bm{x}|\bm{y}) \geq 2\e^{nr} $ is implied by $\hat{H}(\bm{x}|\bm{y}) \geq r + \frac{1}{n} \log 2$.
In what follows, we let $\delta_n' =  \frac{1}{n} \log 2$. We now proceed as
\begin{align}
  &\Prb \big[ \mathcal{A}_1 | \bm{X} = \bm{x} \big] \nonumber\\*
  & \geq \Prb \big[ \hat{H}(\bm{x}|\bm{Y}) \geq r + \delta_n'  \big] \\
  & = \Prb \big[ \hat{I}(\bm{x} \wedge\bm{Y}) \leq \hat{H}(\bm{x}) - r - \delta_n'   \big] \\
  & = \sum_{V \in \mathcal{V}_n(\mathcal{Y};P_X)} \sum_{\bm{y} \in \mathcal{T}_V(\bm{x})}  W^n(\bm{y}|\bm{x}) \nonumber\\*
  &\qquad\qquad\times\ind \left[  I(P_X,V) \le H(P_X) - r - \delta_n'  \right] \\
  & \geq \max_{V : I(P_X,V) \le H(P_X) - r - \delta_n'} \e^{- n D(V\|W |P_X)-n\delta_n }\\
  & = \e^{- n E_{\mathrm{a}}(r + \delta_n'  ,P_X)-n\delta_n}. \label{eqn:PA1_conv}
\end{align}
Combining the bounds in~\eqref{eqn:PE1_conv} and~\eqref{eqn:PA1_conv} and using a similar closing argument as that for the achievability part completes the proof of the ensemble converse.

\section{Ensemble-Tight Strong Converse Exponents} \label{sec:strong_conv_exp}
In this section, we consider the regime in which the ensemble  probability of correct decoding $1-\varepsilon_n$ decays exponentially fast with the blocklength. The exponential rate of decay of $1-\varepsilon_n$ is known as  the {\em strong converse exponent}~\cite{arimoto73,DK79}. This regime is of interest, in usual channel coding with arbitrary decoders and in the absence of abandonment, when the code rate is {\em above} the channel capacity $C(W)$. The main result in this section (Theorem~\ref{thm:sce}) implies that the strong converse exponent is positive if {\em either} the code rate is larger $C(W)$ {\em or} the abandonment rate $r$ is smaller than $H_{P_X\times W}(X|Y)$.

For fixed $R,r,P_X$, and $W$, we define 
\begin{align}
    K_{\mathrm{r}}(R,P_X)  &:= \min_{V}D(V\|W|P_X)+ |R-I(P_X,V) |^+ ,\\*
    K_{\mathrm{sp}}(R,P_X )  &:= \min_{V: I(P_X,V)\ge R}D(V\|W|P_X) . \label{eqn:defKsp}
\end{align}
The subscripts $\mathrm{r}$ and $\mathrm{sp}$ allude  to the fact that these two functions are the strong converse counterparts of the \underline{r}andom coding and \underline{s}phere-\underline{p}acking error exponents. \red{It is clear that $K_{\mathrm{r}}(R,P_X)  \leq K_{\mathrm{sp}}(R,P_X )$.   Arimoto~\cite{arimoto73} demonstrated that the strong converse exponent for the DMC $W$ is at least $\min_{P_X} K_{\mathrm{r}}(R,P_X)$, proving this converse result in an equivalent dual form. Later, Omura~\cite{Omura1975} proved that $\min_{P_X} K_{\mathrm{sp}}(R,P_X)$  is achievable, thereby establishing an upper bound on the strong converse exponent. Subsequently, Dueck and K\"orner~\cite{DK79} refined Omura's approach by employing a code extension argument, showing that $\min_{P_X} K_{\mathrm{r}}(R,P_X)$ is also achievable. This result ultimately established the strong converse exponent for the DMC $W$ for all rates $R$.}


\red{It was claimed by Oohama~\cite[first part of Property~3(b)]{oohama15}  that $K_{\mathrm{r}}(R,P_X)=K_{\mathrm{sp}}(R,P_X)$ for all $0\le R\le \log|\mathcal{X}|$.  However, the proof in~\cite{oohama15} contains a gap. In particular, it is claimed\footnote{The error in the first part of Property~3(b) in~\cite{oohama15} was confirmed by Prof.\ Oohama to the authors in a private communication.} that the function $g(V) :=  -I(P,V) + D(V\|W|P)$ is {\em linear}. Unfortunately, one can only verify that  $g(V)$ is convex but not necessarily linear.} 

In analogy to Definitions~\ref{def:2nd_order} and~\ref{def:err_exp}, we define the ensemble strong converse exponent as follows.
\begin{definition} \label{def:sc_exp}
Fixed a limiting distribution $P_X$. The strong converse exponent $K\in\mathbb{R}_+$ is  {\em achievable} at rate pair $(R,r)$ if the sequence of $(n,R,r, P_{X}^{(n)})$-codes with \red{random-coding} error probabilities $\{\varepsilon_n\}_{n\in\mathbb{N}}$ satisfies  
\begin{align}
    \limsup_{n\to\infty}\frac{1}{n}\log\frac{1}{1-\varepsilon_n}&\le K. \label{eqn:ach_sc_exp}
\end{align}
The infimum of all achievable strong converse exponents for a given rate pair $(R,r)$, limiting input distribution $P_X$ and DMC~$W$ is denoted as $K^*(R,r,P_X)$. Minimizing  $K^*(R,r,P_X)$ over $P_X$ yields the {\em minimium ensemble strong converse exponent} $K^*(R,r)$.
\end{definition}
We now present the strong converse exponent result.
\begin{theorem}\label{thm:sce}
    For any $P_X \in \mathcal{P}(\mathcal{X})$ and $(R,r) \in \mathbb{R}^2_{+}$,
    \begin{align}
         K^*(R,r,P_X) =  K_{\mathrm{sp}}(\max\{R,H(P_X) -r\},P_X).\label{eqn:K_star1} 
    \end{align}
    Moreover, $K^*(R,r) = \min_{P_X}K^*(R,r,P_X) $.
\end{theorem}
\red{Similar to Theorem~\ref{thm:ee}, Theorem~\ref{thm:sce} demonstrates that two factors influence the ensemble strong converse exponent: the code rate $R$ and the abandonment rate $r$. The interpretations for the cases $r \geq H(P_X) - R $  and $ r < H(P_X) - R $ are analogous to those discussed following Theorem~\ref{thm:ee}. The primary difference is that these interpretations apply across all rates in the strong converse regime, as there is no critical rate in this scenario~\cite{DK79}.}

\begin{figure}
    \centering
    \includegraphics[width=0.95\linewidth]{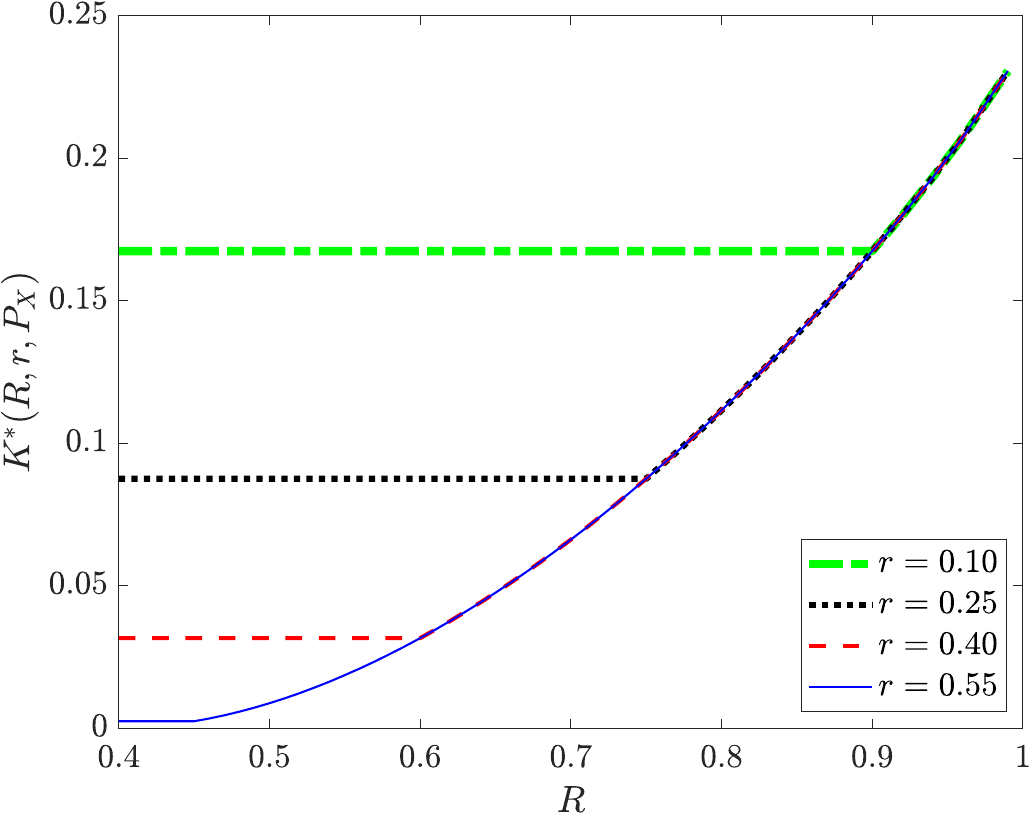}
    \caption{An illustration of $K^*(R,r,P_X)$}
    \label{fig:Kstar}
\end{figure}

\red{If $r \ge H(P_X) - R$, then $K^*(R, r, P_X)$ in~\eqref{eqn:K_star1} simplifies to $K_{\mathrm{sp}}(R, P_X)$. While this resembles $K_{\mathrm{r}}(R, P_X)$, the strong converse exponent for channel coding with a fixed input distribution $P_X$~\cite{DK79}, it is important to note that they are not identical. As noted earlier, we generally have $ K_{\mathrm{r}}(R, P_X) \le K_{\mathrm{sp}}(R, P_X) $. Consequently, when $ r \ge H(P_X) - R $ (implying the code rate is dominant), our ensemble strong converse exponent result is weaker than the classical strong converse exponent established by Dueck and K\"orner~\cite{DK79}. This is because, unlike standard information-theoretic proofs in the strong converse regime, our approach involves bounding the {\em ensemble} average correct decoding probability, which may result in some looseness. A similar situation occurs for rates below capacity, where Gallager showed that the exponent of the ensemble average error probability is the random coding error exponent~\cite{gallager73}.}

\red{If $r<H(P_X)-R$ (implying that the abandonment rate is dominant), then $K^*(R,r,P_X)$ specializes to $K_{\mathrm{sp}}(H(P_X)-r, P_X)$, which coincides with the strong converse exponent for the conditional source coding  setting with source $\bm{X} \sim \mathrm{Unif}\big(\mathcal{T}_{P_X^{(n)}} \big)$ and  side information $\bm{Y} \sim W^n(\cdot | \bm{X})$.
In this case, $K_{\mathrm{sp}}(H(P_X)-r, P_X)$ is equal to  $K_{\mathrm{r}}(H(P_X)-r, P_X) = \min_V \big\{ D(V\|W|P_X)+ | H_{P_X\times V}(X|Y)-r|^+\big\}$. The equivalence between these two exponents was established by  Oohama and Han \cite[Property~3]{oohamaHan94}.}

\begin{example} \red{{\em The strong converse exponent $K^*(R,r,P_X)$ is illustrated in Fig.~\ref{fig:Kstar} for the binary asymmetric channel in~\eqref{eqn:bac} with uniform input distribution. Since the optimization problem defining $K_{\mathrm{sp}}(R,P_X)$ is non-convex, we used a grid search to solve the constrained conditional KL divergence minimization problem in \eqref{eqn:defKsp}. Observe from Fig.~\ref{fig:Kstar} that for a fixed abandonment rate $r$, $K^*(R,r,P_X)$ is initially flat and then increases in accordance with the shape of $K_{\mathrm{sp}}(R,P_X)$. On the other hand, for a fixed code rate $R$, $K^*(R,r,P_X)$ decreases as~$r$ increases. This reflects the fact that the strong converse exponent decreases (improves) when decoding is terminated after more guesses.}} 
\end{example}

\subsection{Proof of the Achievability Part of Theorem~\ref{thm:sce} }
For this part, we are required to {\em lower bound}  the ensemble probability of correct decoding
\begin{align}
    1-\varepsilon_n&= \Prb\left[ \mathcal{E}_1^{\comp}\cap \mathcal{A}_1^{\comp} \right] \\
    &= \sum_{\bm{x},\bm{y} } P_{\bm{X}}(\bm{x})W^n(\bm{y}|\bm{x})   \big(1-\varphi( \bm{x},\bm{y}) \big),
\end{align}
where recall  that $\varphi(\bm{x},\bm{y}):= \Prb \left[ \mathcal{E}_1 \cup \mathcal{A}_1 \mid \bm{X}(1) =\bm{x}, \bm{Y}=\bm{y} \right]$ (see~\eqref{eqn:varphi_def}).  For brevity, we write $\overline{\varphi}(\bm{x},\bm{y}):=1-\varphi( \bm{x},\bm{y})$.

It is easy to see  that $(\bm{x}, \bm{y})\notin\mathcal{S}_r$ (recall the definition of the set $\mathcal{S}_r$ in \eqref{eqn:def_Sr}),  $\overline{\varphi}(\bm{x},\bm{y})=0$. Hence,
\begin{align}
    1-\varepsilon_n=\sum_{ (\bm{x},\bm{y} )\in \mathcal{S}_r } P_{\bm{X}}(\bm{x})W^n(\bm{y}|\bm{x})  \overline{\varphi}(\bm{x},\bm{y}). \label{eqn:sc1}
\end{align}
We now lower bound $\overline{\varphi}(\bm{x},\bm{y})$ assuming that $(\bm{x},\bm{y} )\in \mathcal{S}_r$. Using the independence of the codewords $\bm{X}(w),w\in [M_n]$, 
\begin{align}
    &\overline{\varphi}(\bm{x},\bm{y}) \nonumber\\*
    &= \! \prod_{w\ne 1} \! \Prb\Big[\big\{ G(\bm{X}(w) | \bm{Y})\! \ge\!  G(\bm{X}(1) | \bm{Y}) \big\} \Big| \bm{X}(1) \! =\!\bm{x} ,\bm{Y}\! =\!\bm{y}\Big]\!  \\ 
     &= \big(  1-\Psi(\bm{x},\bm{y}) \big)^{M_n-1}.\label{eqn:sc2}
\end{align}
Define the set 
\begin{align}
    \mathcal{H}_n:= \left\{ (\bm{x},\bm{y}):  \Psi(\bm{x},\bm{y})\le\frac{1}{M_n-1} \right\}.
\end{align}
Then for all $(\bm{x},\bm{y} )\in \mathcal{S}_r \cap\mathcal{H}_n$, we have 
\begin{align}
    \!\!\big(  1\!-\!\Psi(\bm{x},\bm{y}) \big)^{M_n-1} \ge \Big( 1\!-\!\frac{1}{M_n\!-\! 1} \Big)^{M_n-1} \!\to\! \e^{ -1}\approx 0.37.
\end{align}
From the above calculations, for $n$ large enough, we have 
\begin{align}
    &1-\varepsilon_n \nonumber\\*
    & \ge\sum_{(\bm{x},\bm{y} )\in \mathcal{S}_r \cap\mathcal{H}_n} P_{\bm{X}}(\bm{x})W^n(\bm{y}|\bm{x}) \overline{\varphi}(\bm{x},\bm{y}) \\
    &\ge\frac{1}{4} \Prb\left[(\bm{X},\bm{Y})\in  \mathcal{S}_r \cap\mathcal{F}_n \right]\\
    &\ge\frac{1}{4}\Prb\bigg[ \left\{ G(\bm{X}|\bm{Y})\le \e^{nr}\right\}\cap \Big\{ \Psi( \bm{X},\bm{Y} ) \le\frac{1}{M_n-1}\Big\} \bigg]\\
    &\ge\frac{1}{4} \sum_{\bm{x}\in\mathcal{T}_{P_X}} \frac{1}{| \mathcal{T}_{P_X} |} \sum_{V\in \mathcal{V}_n(\mathcal{Y};P_X) }W^n(\mathcal{T}_V( \bm{x}) |\bm{x})  \nonumber\\*
    &\qquad\times \mathbbm{1}\big[     I (P_X,V)\ge\max\{ H(P_X)-r , R\} +2\delta_n \big]  \label{eqn:use_bds}\\
    &\ge\frac{1}{4}\e^{  -n K_{\mathrm{sp}}( \max\{ R,H(P_X)-r \} +2\delta_n,P_X)},
\end{align}
where \eqref{eqn:use_bds} uses the upper bound on $G(\bm{x}|\bm{y})$ stated in \eqref{eq:G_bounds} and the lower bound on $\Psi(\bm{x},\bm{y})$ stated in Lemma \ref{lem:bound_psi}. Thus, we obtain 
\begin{align}
    \limsup_{n\to\infty}\frac{1}{n}\log\frac{1}{1-\varepsilon_n}\le K_{\mathrm{sp}}(\max\{R, H(P_X)-r\},P_X).
\end{align}
\subsection{\red{Proof of the Ensemble Converse Part of Theorem~\ref{thm:sce} }}
\label{subsec:converse_sce}
For this part, we first define 
\begin{align}
    R'&:= R-\frac{2 (|\mathcal{X} ||\mathcal{Y}|+1)\log(n+1)}{n}\label{eqn:def_Rprime} \quad\mbox{and}\\
    \mathcal{J}_{R'}&:= \big\{(\bm{x},\bm{y} ): \hat{I}(\bm{x}\wedge\bm{y} )\ge R' \big\}.
\end{align}
Recall the definition of the set $\mathcal{S}_r$ in \eqref{eqn:def_Sr}. From~\eqref{eqn:sc1} and~\eqref{eqn:sc2}, we note that the ensemble probability of correct decoding can be expressed as
\begin{align}
&1-\varepsilon_n= \sum_{(\bm{x},\bm{y} )\in \mathcal{S}_r }P_{\bm{X}}(\bm{x})W^n(\bm{y} | \bm{x}) \big( 1 - \Psi(\bm{x},\bm{y}) \big)^{M_n - 1}. 
\end{align}
This sum can be split into  
\begin{align}
(\vardiamond):=\sum_{(\bm{x},\bm{y} )\in \mathcal{S}_r \cap\mathcal{J}_{R'} }\!P_{\bm{X}}(\bm{x})W^n(\bm{y} | \bm{x})\big( 1 -\Psi(\bm{x},\bm{y})\big)^{M_n - 1} \!, 
\end{align}
and
\begin{align}
(\varheart) := \sum_{(\bm{x},\bm{y} )\in \mathcal{S}_r \cap\mathcal{J}_{R'}^{\comp} }\!P_{\bm{X}}(\bm{x})W^n(\bm{y} | \bm{x})\big( 1 -\Psi(\bm{x},\bm{y})\big)^{M_n - 1} \!.
\end{align}
We first consider $(\varheart)$, which can be upper bounded by summing only over those $(\bm{x},\bm{y} )$ in $\mathcal{J}_{R'}^{\comp}$ (dropping $\mathcal{S}_r $), i.e., 
\begin{align}
    (\varheart) & \le \sum_{(\bm{x},\bm{y} )\in   \mathcal{J}_{R'}^{\comp} } P_{\bm{X}}(\bm{x})W^n(\bm{y} | \bm{x})\big( 1 -\Psi(\bm{x},\bm{y})\big)^{M_n - 1} \! \\
    &\le\sum_{(\bm{x},\bm{y} )\in   \mathcal{J}_{R'}^{\comp} } \! P_{\bm{X}}(\bm{x})W^n(\bm{y} | \bm{x})\bigg( 1 \! -\! \frac{\e^{-n\hat{I}(\bm{x}\wedge\bm{y})}}{(n\! +\! 1)^{2|\mathcal{X}||\mathcal{Y}| }} \bigg)^{M_n - 1}, \! \label{eqn:bd_psi}
\end{align}
where  \eqref{eqn:bd_psi} follows from the lower bound on $\Psi(\bm{x},\bm{y})$ stated in Lemma \ref{lem:bound_psi}.
By using the fact that the code has constant composition $P_X$, 
\begin{align}
    \!(\varheart)\!\le\! \sum_{V:  I(P_X,V) < R'} \!\!\e^{-nD(V\|W|P_X)}\left( 1 \! -\!  \frac{\e^{-n I (P_X,V )}}{(n\!+\!1)^{ 2|\mathcal{X} ||\mathcal{Y}| }} \right)^{M_n - 1} \!\! .
\end{align}
Furthermore, for $n$ large enough, $M_n-1\ge\frac{1}{2} \e^{nR}$. Combining this with the inequality $(1-x)^k\le\exp(-kx)$, we obtain 
\begin{align}
  (\varheart)& 
  \le   \!\! \sum_{V : I(P_X,V) < R'} \!\!\!\e^{-nD(V\|W|P_X)}\exp\bigg( \! -\! \frac{\e^{-nI(P_X,V )}}{  (n \!+\! 1)^{ 2|\mathcal{X} ||\mathcal{Y}| }}  \frac{\e^{nR}}{2} \bigg) \\
  & \le  \!\!  \sum_{V : I(P_X,V) < R'} \!\!\!\e^{-nD(V\|W|P_X)}\exp\bigg(\!- \!\frac{\e^{-nR' }}{  (n \!+\! 1)^{ 2|\mathcal{X} ||\mathcal{Y}| }}  \frac{\e^{nR}}{2} \bigg)\label{eqn:use_Rprime}\\
  &\le \sum_{V : I(P_X,V) < R'}  \!\!\e^{-nD(V\|W|P_X)}\exp\bigg(-\frac{1}{2}(n\!+\!1)^2 \bigg) , \label{eqn:double_exp}
\end{align}
where~\eqref{eqn:double_exp} uses the definition of $R'$ in \eqref{eqn:def_Rprime}. Thus, $(\spadesuit)$ decays faster than any exponential in $n$. 

The term $(1-\Psi(\bm{x},\bm{y}) )^{M_n-1}$ in $(\vardiamond)$  can be trivially upper bounded by $1$. This means that 
\begin{align}
(\vardiamond) &\le   \sum_{(\bm{x},\bm{y} )\in \mathcal{S}_r \cap\mathcal{J}_{R'} }\!\!P_{\bm{X}}(\bm{x})W^n(\bm{y} | \bm{x}) \\
&\le \sum_{V\in \mathcal{V}_n(\mathcal{Y};P_X)}\e^{-nD(V\|W|P_X )}\times \mathbbm{1}\left[I(P_X,V)\ge R'\right]  \\
&\qquad\qquad\times \mathbbm{1}\left[H_{P_X\times V}(X|Y)\le r\right] ,\label{eqn:bd_diam}
\end{align}
where in~\eqref{eqn:bd_diam}, we used the fact that $G(\bm{x}|\bm{y})\le \e^{nr}$ implies that $\hat{H}(\bm{x}|\bm{y})\le r$. 

Combining \eqref{eqn:double_exp} and~\eqref{eqn:bd_diam}, we obtain 
\begin{align}
    \!1-\varepsilon_n &\le 2 \sum_{V\in \mathcal{V}_n(\mathcal{Y};P_X)}\!\e^{-nD(V\|W|P_X )}\times \mathbbm{1}\left[I(P_X,V)\!\ge\! R'\right]  \\
&\qquad\qquad\times \mathbbm{1}\left[H_{P_X\times V}(X|Y)\le r\right] 
\end{align}
for all $n$ sufficiently large. This implies that 
\begin{align}
    \liminf_{n\to\infty}\frac{1}{n}\log\frac{1}{1-\varepsilon_n}\ge K_{\mathrm{sp}}(\max\{R, H(P_X)-r\},P_X).
\end{align}

\section{Conclusion and Future Work}
This paper elucidates the interplay between the code and abandonment rates in terms of the second-order, error exponent and strong converse exponent asymptotics for guessing-based decoders with abandonment. The second-order results suggest that unless the two ``backoff terms'' $s$ and $t$ are equal, either (but not both) the effect of regular decoding or abandonment dominates the ensemble error probability if it is designed to be non-vanishing. The exponent results underscore  the dichotomy between the code rate and abandonment rate. 

For future work, one can consider extending the current results to other families of channels such as additive white Gaussian noise channels. Additionally, a more ambitious direction involves establishing a ``partial information-theoretic converse'' of the following form: consider all decoding strategies where  all  input sequences are ranked according to a metric  $q : \mathcal{X}^n\times\mathcal{Y}^n\to\mathbb{N}$ (using, for example, the metric $ G(\bm{x}|\bm{y}) $ in this work), but the search is terminated after $m_n=\e^{n r_n}$ guesses. An important question is to determine the fundamental limits of such schemes in various asymptotic regimes, especially when we are permitted to optimize over both the code design   and the decoding metric $q$. This is in contrast with our specific setup of constant composition codes and empirical entropy-based ranking with abandonment.
\appendices
\section{Proof of Lemma~\ref{lem:bound_psi} } \label{app:prf_bound_psi}
\begin{proof} \red{In this proof, we make   use of the following well known bounds on the size of the type class and the $V$-shell~\cite[Lemma~2.5]{Csi97}:
\begin{align}
   (n\!+\!1)^{-|\mathcal{X}|}\e^{nH(P_X)} &\le |\mathcal{T}_{P_X}|\le \e^{nH(P_X)},\\
   (n\!+\!1)^{-|\mathcal{X}||\mathcal{Y}|  } \e^{nH(V_{X|Y}|\hat{P}_{\bm{y}})}  &\!\le\! |\mathcal{T}_{V_{X|Y}}(\bm{y}) |  \!\le\! \e^{nH(V_{X|Y}|\hat{P}_{\bm{y}})} . \label{eqn:bd_Vshell}
\end{align}}
For the upper bound on $\Psi(\bm{x}, \bm{y})$, we observe from \eqref{eq:G_bounds} that 
$G(\overline{\bm{x}} | \bm{y}) \le G(\bm{x}|\bm{y})$ implies $\hat{H}( \overline{\bm{x}} | \bm{y}) \leq \hat{H}( \bm{x} | \bm{y}) + \delta_n$, where 
$\delta_n = \frac{|\mathcal{X}||\mathcal{Y}|}{n}\log(n+1) $.
Therefore, we write 
\begin{align}
&\Psi(\bm{x}, \bm{y})\nonumber\\*
& \leq 
\Prb \big[ \hat{H}(\overline{\bm{X}} | \bm{y}) \leq \hat{H}(\bm{x}|\bm{y}) + \delta_n \big]  \\
& = \frac{1}{|\mathcal{T}_{P_X}|} \sum_{\overline{\bm{x}} \in  \mathcal{T}_{P_X} } \ind \big[ \hat{H}( \overline{\bm{x}} | \bm{y}) \leq \hat{H}(\bm{x} | \bm{y}) + \delta_n \big] \\
& = \frac{1}{|\mathcal{T}_{P_X}|} \!  \sum_{V_{X|Y}: V_{X|Y}\hat{P}_{\bm{y}} = P_X} \sum_{\overline{\bm{x}} \in \mathcal{T}_{V_{X|Y}}(\bm{y}) } \!\!\ind \big[ \hat{H}( \overline{\bm{x}} | \bm{y}) \!\leq\! \hat{H}(\bm{x} | \bm{y}) + \delta_n \big] \\
& = \frac{1}{|\mathcal{T}_{P_X}|} \sum_{\substack{ V_{X|Y}: V_{X|Y}\hat{P}_{\bm{y}} = P_X, \\ H(V_{X|Y}|\hat{P}_{\bm{y}}) \leq \hat{H}(\bm{x}|\bm{y}) + \delta_n } }|\mathcal{T}_{V_{X|Y}}(\bm{y})| \\
& \leq (n\!+\!1)^{|\mathcal{X}||\mathcal{Y}|}\e^{-nH(P_X)}\!\!\!\! \sum_{\substack{ V_{X|Y}: V_{X|Y}\hat{P}_{\bm{y}} = P_X, \\ H(V_{X|Y}|\hat{P}_{\bm{y}}) \leq \hat{H}(\bm{x}|\bm{y}) + \delta_n }} \!\!\!\! \!\e^{nH(V_{X|Y}|\hat{P}_{\bm{y}})} \\
& \leq  (n+1)^{2|\mathcal{X}||\mathcal{Y}|} \e^{-n\hat{I}(\bm{x}\wedge\bm{y}) + n \delta_n}.
\end{align}

The lower bound is obtained in a similar fashion by noting, again from \eqref{eq:G_bounds}, that $G( \overline{\bm{x}} | \bm{y}) \le G(\bm{x} | \bm{y})$ is implied by the condition  $\hat{H}( \overline{\bm{x}} | \bm{y}) + \delta_n \le\hat{H}(\bm{x} | \bm{y})$. Therefore,
\begin{align}
&\Psi(\bm{x}, \bm{y})\nonumber\\*
& \geq \frac{1}{|\mathcal{T}_{P_X}|} \sum_{\overline{\bm{x}} \in \mathcal{T}_{P_X} } \ind \big[ \hat{H}( \overline{\bm{x}} | \bm{y}) \le\hat{H}(\bm{x} | \bm{y}) - \delta_n \big] \label{eqn:Psi_tobound}\\
& = \frac{1}{|\mathcal{T}_{P_X}|} 
\sum_{V_{X|Y} : V_{X|Y} \hat{P}_{\bm{y}} = P_X}
\sum_{\overline{\bm{x}} \in \mathcal{T}_V(\bm{y}) } \ind \big[ \hat{H}( \overline{\bm{x}} | \bm{y}) \le \hat{H}(\bm{x} | \bm{y}) - \delta_n \big] \\
& = \frac{1}{|\mathcal{T}_{P_X}|} 
\sum_{\substack{ V_{X|Y} : V_{X|Y} \hat{P}_{\bm{y}} = P_X,  \\H(V_{X|Y}|\hat{P}_{\bm{y}}) \le \hat{H}(\bm{x}|\bm{y}) - \delta_n } } |\mathcal{T}_V(\bm{y})| \\
& \geq \frac{1}{|\mathcal{T}_{P_X}|} \!\!
\sum_{ \substack{ V_{X|Y} : V_{X|Y} \hat{P}_{\bm{y}} = P_X,  \\H(V_{X|Y}|\hat{P}_{\bm{y}}) \le \hat{H}(\bm{x}|\bm{y}) - \delta_n} }\!\!
(n+1)^{- |\mathcal{X}||\mathcal{Y}|} \e^{n H(V_{X|Y}|\hat{P}_{\bm{y}})} \\
& \geq \frac{1}{|\mathcal{T}_{P_X}|} 
(n+1)^{- |\mathcal{X}||\mathcal{Y}|} \!\!\max_{ \substack{ V_{X|Y} : V_{X|Y} \hat{P}_{\bm{y}} = P_X, \\ H(V_{X|Y}|\hat{P}_{\bm{y}}) \le \hat{H}(\bm{x}|\bm{y}) - \delta_n} }\!\! \e^{n H(V_{X|Y}|\hat{P}_{\bm{y}})} \\
& \geq 
(n+1)^{- |\mathcal{X}||\mathcal{Y}|} \e^{-n \hat{I}(\bm{x} \wedge \bm{y}) - n\delta_n}. \label{eqn:compareG}
\end{align}
This completes the proof of Lemma \ref{lem:bound_psi}.
\end{proof} 

\paragraph*{Acknowledgments}
We gratefully acknowledge discussions with Prof.~Y.~Oohama and Prof.~B.~Nakiboğlu. In particular, Prof.\ Nakiboğlu identified an error in the proof of Theorem~\ref{thm:sce} in an earlier version of this paper.

\bibliographystyle{IEEEtran}
\bibliography{References}

\begin{IEEEbiographynophoto}{Vincent Y. F. Tan} (Senior Member, IEEE) was born in Singapore, in 1981.
He received the B.A.\ and M.Eng.\ degrees in electrical and information
science from Cambridge University in 2005 and the Ph.D degree in electrical
engineering and computer science (EECS) from Massachusetts Institute of
Technology (MIT) in 2011. He is currently a Professor with the Department of Mathematics and the
Department of Electrical and Computer Engineering (ECE), National University of Singapore (NUS). His research interests include information theory,
machine learning, and statistical signal processing. 

He is an elected member
of the IEEE Information Theory Society Board of Governors. He was an
IEEE Information Theory Society Distinguished Lecturer from 2018 to 2019.
He received the MIT EECS Jin-Au Kong Outstanding Doctoral Thesis Prize
in 2011, the NUS Young Investigator Award in 2014, Singapore National
Research Foundation (NRF) Fellowship (Class of 2018), and the NUS Young
Researcher Award in 2019. He is also serving as a Senior Area Editor for
{\em IEEE Transactions on Signal Processing} and an Area Editor in
Shannon Theory and Information Measures for {\em  IEEE Transactions on  Information Theory}. He also regularly serves as an  Area Chair or Senior Area Chair for prominent
Machine Learning conferences, such as the International Conference on
Learning Representations (ICLR) and the Conference on Neural Information
Processing Systems (NeurIPS).
\end{IEEEbiographynophoto}

\begin{IEEEbiographynophoto}{Hamdi Joudeh} (Member, IEEE) received the M.Sc.\ and Ph.D.\ degrees in electrical engineering from Imperial College London in 2011 and 2016, respectively. From 2016 to 2020, he held research positions at Imperial College London and Technische Universität Berlin. He is currently an Associate Professor in the Department of Electrical Engineering at Eindhoven University of Technology (TU/e). His research interests include information theory, communications and signal processing.

Dr. Joudeh received a Starting Grant from the European Research Council (ERC) in 2023. He served as an Associate Editor for \emph{IEEE Transactions on Signal Processing} and for \emph{IEEE Communications Letters}.
\end{IEEEbiographynophoto}

\end{document}